\documentclass[12pt]{article}    
\pdfoutput=1

\usepackage[vmargin={1in, 1in},hmargin={1in, 1in}]{geometry}

\usepackage{latexsym}
\usepackage{amssymb}
\usepackage{amsmath}
\usepackage{amsthm}
\usepackage{mathtools}
\usepackage{amsfonts}
\usepackage{graphicx}
\usepackage{enumerate}
\usepackage{xcolor}
	\definecolor{MyDarkBlue}{rgb}{0,0.08,0.45}
\usepackage{hyperref}
	\hypersetup{linkcolor=MyDarkBlue, citecolor=MyDarkBlue, colorlinks=true}
\usepackage{setspace}
\usepackage{booktabs}
\usepackage{pgfplots}
\pgfplotsset{compat=1.18}
\usetikzlibrary{arrows.meta}
\usetikzlibrary{positioning}

\usepackage{natbib}
\newcommand\cites[1]{\citeauthor{#1}'s \citeyearpar{#1}}
\usepackage{hypernat}
\usepackage{booktabs}
\usepackage{caption}
\usepackage{bbm}

\usetikzlibrary{calc,arrows.meta}

\definecolor{paperred}{RGB}{255,55,55}
\definecolor{papergreen}{RGB}{75,115,0}

\tikzset{
  simplex/.style={draw=black,line width=0.9pt},
  redline/.style={draw=paperred,line width=0.45pt},
  greenline/.style={draw=papergreen,line width=0.45pt},
  vector/.style={->,>=Latex,draw=black,line width=1pt},
  pt/.style={circle,fill=black,inner sep=1.5pt},
  lab/.style={font=\footnotesize,inner sep=1pt},
  region/.style={font=\large\color{papergreen},inner sep=1pt}
}

\newtheorem{theorem}{Theorem}
\newtheorem*{theorem*}{Theorem}

\newtheorem{proposition}{Proposition}
\newtheorem{lemma}{Lemma}
\newtheorem{corollary}{Corollary}

\theoremstyle{definition}
\newtheorem{definition}{Definition}

\newtheorem{condition}{Condition}
\newtheorem{example}{Example}
\newtheorem{axiom}{Axiom}
\newtheorem*{axiom*}{Axiom}

\newcommand{\finex}{\leavevmode\unskip\penalty9999 \hbox{}\nobreak\hfill\quad\hbox{$\blacktriangle$}}
\newcommand{\finexhere}{\begingroup \let\mathqed\math@qedhere
    \let\qed@elt\setQED@elt \blacktriangle@stack\bRelax\bRelax \endgroup}

\newcommand{\cD}{\mathcal{D}}

\newcommand{\cU}{\mathcal{U}}

\newcommand{\bE}{\mathbb{E}}

\newcommand{\bP}{\mathbb{P}}

\newcommand{\bR}{\mathbb{R}}
\newcommand{\bZ}{\mathbb{Z}}

\newcommand{\supp}{\mathrm{supp\,}}

\newcommand{\ext}{\mathrm{ext\,}}

\newcommand{\fosd}{\geq_{\textrm{FOSD}}}

\newcommand{\Fosd}{>_{\textrm{FOSD}}}

\newcommand{\conv}{\mathrm{conv\,}}
\newcommand{\cone}{\mathrm{cone\,}}
\newcommand{\cones}{\mathrm{cone}_s}

\DeclareMathOperator{\1}{\mathbbm{1}}
\newcommand{\spimp}[0] {\hspace{10pt}  \implies \hspace{10pt}}
\newcommand{\spiff}[0] {\hspace{10pt}  \iff \hspace{10pt}}
\DeclareMathOperator{\spann}{span}
\DeclareMathOperator{\rank}{rank}

\expandafter\def\expandafter\normalsize\expandafter{%
    \normalsize%
    \setlength\abovedisplayskip{4pt}%
    \setlength\belowdisplayskip{4pt}%
    \setlength\abovedisplayshortskip{4pt}%
    \setlength\belowdisplayshortskip{4pt}%
}

\title{\vspace*{-2em}Axioms and Anomalies with Finite Data\thanks{We thank Joyee Deb, Ryota Iijima, Lasse Mononen, Ashesh Rambachan, and Weijie Zhong, the audience at the MIT-Harvard theory seminar, and Refine.ink for helpful comments. We thank Josh Peterson and Tom Griffiths for sharing the full individual-level \texttt{choices13k} dataset.}}
\author{Cheaheon Lim and Tomasz Strzalecki\thanks{Department of Economics, Harvard University. Date: \today.}}
\date{}

\begin{document}
	
\maketitle
\vspace*{-1em}
\begin{abstract}
	The classical expected utility (EU) axioms are not sufficient for finite datasets. There are a number of anomalies (violations of EU) where axioms are satisfied. This paper studies axioms that are immune to this problem and definitively delineate between EU and non-EU. We discuss implications for experimental design and explore the automatic generation of anomalies.
\end{abstract}

\section{Introduction}   
This paper studies situations where the analyst observes a finite dataset of pairwise choices between lotteries. There  are datasets of this form for which there is a gap between axioms and the representation: all of the \cite{vNM} axioms are satisfied, but there is no expected utility (EU) representation. The \cite{Allais53} paradox and more generally \cites{Machina87} common consequence effect  are examples of such datasets. In this sense, the classical axioms of decision theory are not sufficient: simply checking the axioms on the dataset does not detect all the possible violations of the theory. We explain this in detail in Section \ref{sec:setup}.

This problem was studied by \cite{Fishburn75}, who proposed a necessary and sufficient condition for EU on finite datasets. In Section \ref{sec:EU} we review this approach and develop our strong independence axiom. Based on this framework, in Section \ref{sec:REU}, we focus on the setting where recorded choices are stochastic (such as when the experimenter only reports population averages, not individual choice data). The elegant axiomatic characterization of \cite{GP06} requires an infinite amount of data. We develop axioms that are necessary and sufficient on small  datasets (our Theorem \ref{thm:REU}). 

Our work has several implications for experimental work, which we discuss in Section \ref{sec:applications}. First, some researchers purposely tie their hands and draw lottery pairs continuously at random.\footnote{For example, the experimental protocol of \cite{Erevetal17} randomly draws lottery pairs. The largest collection of pairwise choices between lotteries, the \texttt{choices13k} dataset, follows this protocol \citep{13k}.} Our Corollary \ref{cor:random-erev} shows that on such datasets, EU is satisfied with probability one (irrespective of the actual choices of the subjects). Thus, they are not useful for testing EU. The experimenter needs to design the experiment carefully to set up special traps for the subject, as \cite{Allais53} and \cite{KT79} did.\footnote{We also show that adding the requirement that the Bernoulli utility be monotone or concave does have bite even on randomly designed experiments. But we note that the classic anomalies such as the Allais paradox and certainty effect are violations of pure EU and cannot be detected using this data.}

Second, experimental economists often use a model of choice which is EU plus logit shocks. It is well known that such models can spuriously   generate Allais-style anomalies that are normally associated with non-EU preferences.  This is a confound that is taken seriously in the literature  because it means that we cannot distinguish truly non-EU behavior from EU plus noise \citep{ballinger1997decisions,Loomes05,hey2005we,blavatskyy2007stochastic,Wilcox08,blavatskyy2010reverse,bhatia2017noisy,Mcgranaghan24}. We show that testing a version of our strong linearity axiom entirely removes this confound. We design an experiment where the only violations of the axiom come from true non-EU behavior and the axiom is satisfied both for regular EU and EU plus logit noise agents.

We are also motivated by the recent development of  computational methods to automatically generate and classify anomalies  \citep{MR}. We show how our axioms relate to such anomalies, and  develop computational methods which find many anomalies that escape their classification. We think of our method as complementary to theirs because we impose a grid on the space of outcomes, whereas they search the bigger space of distributions. 

Finally, a common belief is that we need axioms because directly testing a theory is computationally impossible due to the infinite number of utility functions to check. Yet, we  illustrate that EU theory can be tested directly   in polynomial time. This is done exactly, without any approximations such as splines. Moreover, we show that on finite datasets it is actually easier to check the representation directly, rather than checking the axioms.

In the Appendix, we show how to extend our methods to axiomatize probability weighting, the key component of  \cites{KT79} prospect theory. We also discuss axioms for EU with increasing and/or concave utility. All proofs are in the Appendix. 


Our work connects to the literature on the empirical content of finite datasets \citep{Chambersetal14,mononen2024empirical}, which assume that the analyst observes both the strict and weak preferences. We do not allow for indifferences because then any finite dataset of pairwise choices can be rationalized (by complete indifference). This paper also relates to  the literature on incomplete preferences between lotteries (because a finite dataset can be seen as an incomplete relation). However, the classic independence axiom  \citep{aumann1962utility,DMR} is very demanding and implies that the dataset is infinite.


%
%
\section{Setup}\label{sec:setup}
Let  $\Delta(Z)$ be the set of lotteries (probability distributions) over some finite set $Z$.  We will use the common notation for mixing lotteries $\alpha p + (1-\alpha) r$, where $p, r\in \Delta(Z)$ and $\alpha\in [0, 1]$.

\begin{axiom}[Independence]\label{ax:ind}
	For all $\alpha\in (0,1)$ and $r\in \Delta(Z)$, if $p^1\succ q^1$ and
	\begin{align*}
		p^2&=\alpha p^1 + (1-\alpha) r\\
		q^2&=\alpha q^1 + (1-\alpha) r,
	\end{align*}
	then $p^2\succ q^2$.
\end{axiom}

The celebrated von Neumann–Morgenstern theorem says that if choices between all pairs of lotteries are observed, then independence (alongside weak order and continuity) guarantees that choices are represented by expected utility, 
$$U(p)=\sum_{z\in Z}u(z)p(z)$$
for some $u:Z\to \bR$. However, if the analyst only observes a finite number of choices, these axioms are not strong enough. A classic example of this is the Allais Paradox. 

\begin{example}[Allais Paradox]\label{ex:Allais}
	Suppose that 
	\begin{align*}
	p^1=\alpha \hat p +(1-\alpha) r^1, &\quad q^1=\alpha  \hat q +(1-\alpha) r^1,\\
	p^2=\alpha \hat p +(1-\alpha) r^2,&\quad  q^2=\alpha  \hat q +(1-\alpha) r^2, 		
	\end{align*}
	for some $\alpha\in (0, 1)$ and $\hat p, \hat q,r^1,r^2\in \Delta(Z)$. 
	
	If the analyst only observes the choices $p^1\succ q^1$ and $q^2\succ p^2$, then the independence axiom is satisfied (vacuously). However, an EU representation does not exist, because if it did, then 
	$$U(p^1)-U(q^1)=\alpha[U(\hat p)-U(\hat q)]=U(p^2)-U(q^2),$$
so these two utility differences would have the same sign. The key is that choices between  $\hat p$ and $\hat q$ are not observed by the analyst and thus not a part of the dataset.

Note that this is a whole family of anomalies. The actual Allais paradox corresponds to specific numerical values of those lotteries. Machina's common consequence effect is a sub-family, where some restrictions are made on $\hat p, \hat q, r^1, r^2$.  \finex
\end{example}

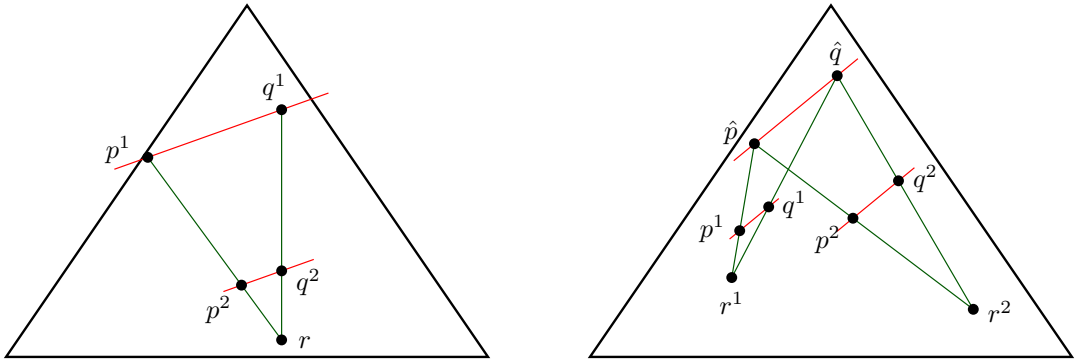
\begin{figure}[h!]
\centering
\begin{tikzpicture}[scale=1.5]
  \tikzset{
    simplex/.style={draw=black,line width=0.9pt},
    redline/.style={draw=red,line width=0.45pt},
    greenline/.style={draw=green!35!black,line width=0.45pt},
    pt/.style={circle,fill=black,inner sep=1.45pt},
    lab/.style={font=\footnotesize,inner sep=1pt}
  }


  \begin{scope}
    \coordinate (A) at (0,0);
    \coordinate (B) at (4.25,0);
    \coordinate (C) at (2.125,3.10);

    \coordinate (R)    at (2.43,0.15);
    \coordinate (Pone) at (1.25,1.76);
    \coordinate (Qone) at (2.43,2.18);
    \coordinate (Ptwo) at ($(Pone)!0.70!(R)$);
    \coordinate (Qtwo) at ($(Qone)!0.70!(R)$);

    \draw[simplex] (A)--(C)--(B)--cycle;

    \draw[redline] ($(Pone)!-0.25!(Qone)$)--($(Pone)!1.35!(Qone)$);
    \draw[redline] ($(Ptwo)!-0.45!(Qtwo)$)--($(Ptwo)!1.80!(Qtwo)$);

    \draw[greenline] (Pone)--(R);
    \draw[greenline] (Qone)--(R);

    \foreach \p in {Pone,Qone,Ptwo,Qtwo,R}
      \node[pt] at (\p) {};

    \node[lab,anchor=east] at ($(Pone)+(-0.12,0.07)$) {$p^1$};
    \node[lab,anchor=west] at ($(Qone)+(-0.2,0.2)$) {$q^1$};
    \node[lab,anchor=north east] at ($(Ptwo)+(-0.06,-0.07)$) {$p^2$};
    \node[lab,anchor=west] at ($(Qtwo)+(0.10,-0.1)$) {$q^2$};
    \node[lab,anchor=west] at ($(R)+(0.12,-0.02)$) {$r$};
  \end{scope}

  \begin{scope}[xshift=5.15cm]
    \coordinate (A) at (0,0);
    \coordinate (B) at (4.25,0);
    \coordinate (C) at (2.125,3.10);

    \coordinate (Phat) at (1.45,1.88);
    \coordinate (Qhat) at (2.18,2.48);
    \coordinate (Rone) at (1.25,0.70);
    \coordinate (Rtwo) at (3.38,0.42);

\coordinate (Pone) at ($(Rone)!0.35!(Phat)$);
\coordinate (Qone) at ($(Rone)!0.35!(Qhat)$);
\coordinate (Ptwo) at ($(Rtwo)!0.55!(Phat)$);
\coordinate (Qtwo) at ($(Rtwo)!0.55!(Qhat)$);

    \draw[simplex] (A)--(C)--(B)--cycle;

    \draw[greenline] (Phat)--(Rone);
    \draw[greenline] (Qhat)--(Rone);
    \draw[greenline] (Phat)--(Rtwo);
    \draw[greenline] (Qhat)--(Rtwo);

    \draw[redline] ($(Phat)!-0.25!(Qhat)$)--($(Phat)!1.25!(Qhat)$);
    \draw[redline] ($(Pone)!-0.35!(Qone)$)--($(Pone)!1.35!(Qone)$);
    \draw[redline] ($(Ptwo)!-0.35!(Qtwo)$)--($(Ptwo)!1.35!(Qtwo)$);

    \foreach \p in {Phat,Qhat,Rone,Rtwo,Pone,Qone,Ptwo,Qtwo}
      \node[pt] at (\p) {};

    \node[lab,anchor=east] at ($(Phat)+(-0.12,0.10)$) {$\hat p$};
    \node[lab,anchor=west] at ($(Qhat)+(-0.1,0.2)$) {$\hat q$};

    \node[lab,anchor=east] at ($(Pone)+(-0.10,0.02)$) {$p^1$};
    \node[lab,anchor=west] at ($(Qone)+(0.09,0.02)$) {$q^1$};

    \node[lab,anchor=north east] at ($(Ptwo)+(-0.08,-0.04)$) {$p^2$};
    \node[lab,anchor=west] at ($(Qtwo)+(0.10,0.02)$) {$q^2$};

    \node[lab,anchor=north] at ($(Rone)+(0,-0.10)$) {$r^1$};
    \node[lab,anchor=west] at ($(Rtwo)+(0.10,-0.02)$) {$r^2$};
  \end{scope}
\end{tikzpicture}
\caption{The lotteries in the independence axiom (left) and the Allais paradox (right).}
\label{fig:independence-allais}
\end{figure}


While the above choices of $p^1$ over $q^1$ and $q^2$ over $p^2$ are not a direct violation of any of the vNM axioms, they violate these axioms \emph{plus their logical implications}. Suppose we wrote down a computer program that formally checks whether each of the vNM axioms holds on a given dataset. This program would say that independence is vacuously satisfied on the dataset above. To rule these choices out based on the vNM axioms, we would need to write down a different computer program that  checks not only  the premise of each axiom, but also the logical deductions that follow from the axioms. That would be a much harder program to write given that there is a continuum of possible deductions to make.\footnote{This would be formally equivalent to checking whether the original incomplete preference can be extended to a full preference over $\Delta(Z)$ that satisfies all the axioms.} To avoid these difficulties, it would be much more convenient to have axioms that can be checked directly, with all possible logical deductions already built in.

In the deterministic case, we say  $D=\big((p_i, q_i, y_i)\big)_{i=1}^n$ is a \emph{dataset} if $p_i\neq q_i$ are both in $\Delta(Z)$, and 
$$y_i=\begin{cases}
	1 &\text{ if }p_i \text{ chosen}\\
	0 &\text{ if }q_i \text{ chosen}.
\end{cases}$$
We interpret $y_i=1$ as $p_i\succ q_i$ as opposed to $p_i\succsim q_i$. According to the latter interpretation, any dataset can be rationalized as  EU with a constant utility function.\footnote{Some experimental protocols allow the agent to express indifference explicitly. In this case we would denote indifference as $y_i=0.5$. We expect that most of our results would extend to this setting and that indifference would add additional uniqueness properties to our representations. We do not follow this route in the paper because most experiments force the agent to make a choice.}

We say that $D$ is \emph{EU-consistent} if there exists $u\in \bR^Z$ such that $$( 2y_i-1) \langle u, p_i-q_i \rangle>0$$ for all $i=1, \ldots, n$, where $\langle \cdot, \cdot \rangle$ denotes the inner product in $\bR^Z$. Otherwise, $D$ is a \emph{counterexample}. We say that $D$ is an \emph{anomaly} if it is a minimal counterexample in the sense of set-inclusion. \cite{MR} call them ``logical anomalies'' because, for many of them, we do not have actual data on human behavior;  one could also call them ``potential anomalies.'' For simplicity, in this paper, we  refer to them as ``anomalies,'' with the hope that future experimental work will close this gap.

In the stochastic case, a \emph{dataset}  is  $D=\big((p_i, q_i, \varrho_i)\big)_{i=1}^n,$ where $p_i\neq q_i$ are both in $\Delta(Z)$, and $\varrho_i\in [0, 1]$ is the fraction of subjects who chose lottery $p_i$ over $q_i$.  We say that $D$ is \emph{REU-consistent} if there exists a Borel probability distribution  $\mu\in \Delta(\bR^Z)$ such that $\varrho_i=\mu(\{u\in \bR^Z: \langle u, p_i-q_i \rangle>0\})$ for all $i=1, \ldots, n$, and ties have $\mu$-zero probability.

\section{Deterministic Choice}\label{sec:EU}

We know that EU has linear and parallel indifference curves. The following axiom captures this idea precisely.

\begin{axiom}[Strong Independence]\label{ax:SI}
	For all $\lambda>0$, if  $p^2-q^2=\lambda (p^1-q^1)$ and $p^1\succ q^1$, then $p^2\succ q^2$.
\end{axiom}


This axiom picks up all possible cases of ``parallel'' in a way that Axiom \ref{ax:ind} does not. For example, consider  the lotteries in Figure \ref{fig:lambda}, where the lines given by $p^1-q^1$ and $p^2-q^2$ are parallel. Independence does not apply here because $r\in \bR^Z$ in Figure \ref{fig:lambda} is outside of the simplex, and  independence therefore holds \emph{vacuously}.  Another case is when $\lambda=1$, such that the two vectors are of equal length, and there does not exist any $r\in \bR^Z$ for which the premise of the independence axiom holds (this case will be important in the sequel).

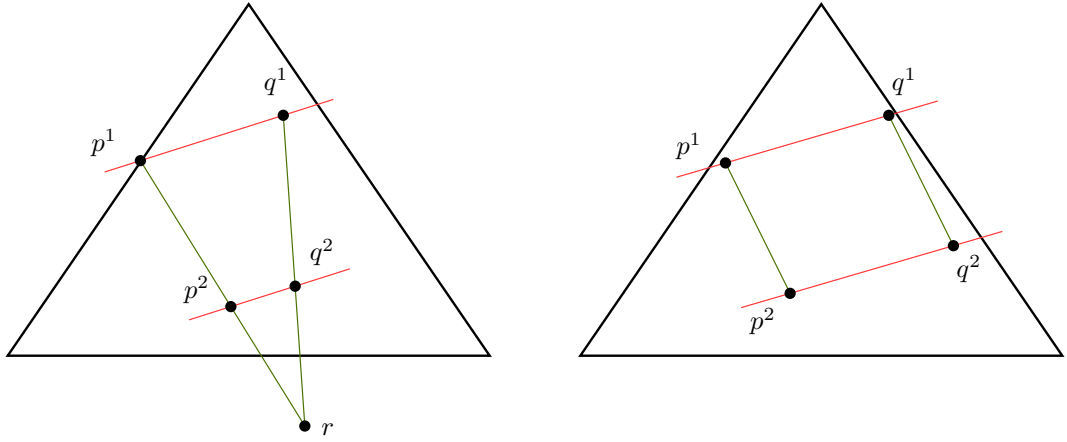
\begin{figure}[h!]
\centering
\vspace{20pt}
\begin{tikzpicture}[scale=1.5]
  \begin{scope}
    \coordinate (A) at (0,0);
    \coordinate (B) at (4.25,0);
    \coordinate (C) at (2.125,3.10);
    \coordinate (R) at (2.62,-0.62);
    \coordinate (Pone) at (1.17,1.72);
    \coordinate (Qone) at (2.43,2.12);
    \coordinate (Ptwo) at ($(R)!0.45!(Pone)$);
    \coordinate (Qtwo) at ($(R)!0.45!(Qone)$);

    \draw[simplex] (A)--(C)--(B)--cycle;
    \draw[redline] ($(Pone)!-0.25!(Qone)$)--($(Pone)!1.35!(Qone)$);
    \draw[redline] ($(Ptwo)!-0.65!(Qtwo)$)--($(Ptwo)!1.85!(Qtwo)$);
    \draw[greenline] (Pone)--(R);
    \draw[greenline] (Qone)--(R);

    \foreach \p in {Pone,Qone,Ptwo,Qtwo,R}
      \node[pt] at (\p) {};

    \node[lab,anchor=east] at ($(Pone)+(-0.18,0.16)$) {$p^1$};
    \node[lab,anchor=west] at ($(Qone)+(-0.2,0.3)$) {$q^1$};
    \node[lab,anchor=east] at ($(Ptwo)+(-0.16,0.15)$) {$p^2$};
    \node[lab,anchor=west] at ($(Qtwo)+(0.1,0.3)$) {$q^2$};
    \node[lab,anchor=west] at ($(R)+(0.12,-0.03)$) {$r$};
  \end{scope}

  \begin{scope}[xshift=5.05cm]
    \coordinate (A) at (0,0);
    \coordinate (B) at (4.25,0);
    \coordinate (C) at (2.125,3.10);
    \coordinate (Pone) at (1.28,1.70);
    \coordinate (Qone) at (2.72,2.12);
    \coordinate (Ptwo) at (1.85,0.55);
    \coordinate (Qtwo) at (3.29,0.97);

    \draw[simplex] (A)--(C)--(B)--cycle;
    \draw[redline] ($(Pone)!-0.3!(Qone)$)--($(Pone)!1.3!(Qone)$);
    \draw[redline] ($(Ptwo)!-0.3!(Qtwo)$)--($(Ptwo)!1.3!(Qtwo)$);
    \draw[greenline] (Pone)--(Ptwo);
    \draw[greenline] (Qone)--(Qtwo);

    \foreach \p in {Pone,Qone,Ptwo,Qtwo}
      \node[pt] at (\p) {};

    \node[lab,anchor=east] at ($(Pone)+(-0.18,0.15)$) {$p^1$};
    \node[lab,anchor=west] at ($(Qone)+(0,0.3)$) {$q^1$};
    \node[lab,anchor=north east] at ($(Ptwo)+(-0.1,-0.1)$) {$p^2$};
    \node[lab,anchor=west] at ($(Qtwo)+(0,-0.20)$) {$q^2$};
  \end{scope}
\end{tikzpicture}
        \caption{The lines given by $p^1-q^1$ and $p^2-q^2$ are parallel. In the left panel $r\notin \Delta(Z)$. In the right panel $\lambda=1$.} \label{fig:lambda}
\end{figure}

While the dataset is indexed by subscripts $D=\{(p_1, q_1, y_1), (p_2, q_2, y_2)\}$, the lotteries in Axioms \ref{ax:ind} and \ref{ax:SI}  are indexed by superscripts because those axioms need to be checked for all possible \textit{enumerations} of the dataset. For example, we could have $(p^1, q^1, y^1)=(p_2, q_2, y_2)$ and $(p^2, q^2, y^2)=(q_1, p_1, 1-y_1)$. In other words, to formally check Axiom \ref{ax:ind} or Axiom \ref{ax:SI}, we need to quantify over all the lottery pairs, regardless of the order they appear in $D$.

For a bigger dataset, ruling out all size-2 anomalies may not be enough to guarantee EU. We need an axiom that rules out anomalies of size-3 and higher.

\begin{axiom}[$k$-Strong Independence]\label{ax:n-SI}$ $
	 If   $p^i\succ q^i$ for $i=1, \ldots, k-1$ and 
\begin{align*}		 
p^k-q^k&=\sum_{i=1}^{k-1} \lambda_i (p^i-q^i), \text{ for }\lambda_i\geq 0,
\end{align*}
then $p^k\succ q^k$.
\end{axiom}

Figure \ref{fig:nSI} illustrates the necessity of this axiom. Under EU, we have $p^i\succ q^i$ if the (normalized) utility function lies in the half space to which $q^i-p^i$ is a normal vector. If $p^1\succ q^1$, then the  utility function must lie in area $A\cup B\cup C$. If $p^2\succ q^2$, then the utility function must lie in area $C\cup D\cup E$. Taken together, the utility function must then lie in area $C$. If $q^3-p^3$ lies between $q^1-p^1$ and $q^2-p^2$, then area $C$ is a subset of the half space to which $q^3-p^3$ is a normal vector (which here is area $B \cup C \cup D$), so it follows that  $p^3\succ q^3$. 

\begin{figure}[h!]
\centering
\vspace{20pt}
\begin{tikzpicture}[scale=1.5]
  \coordinate (O) at (0,0);

  \draw[greenline,line width=0.65pt] (200:3.05)--(20:3.25);
  \draw[greenline,line width=0.65pt] (164:2.95)--(-16:3.10);
  \draw[greenline,line width=0.65pt] (128:2.85)--(-52:2.55);

  \draw[vector] (O)--(112:2.70);
  \draw[vector] (O)--(75:1.55);
  \draw[vector] (O)--(41:2.65);

  \node[lab,anchor=south] at (-0.55,2.5) {$q^2-p^2$};
  \node[lab,anchor=west] at (0.52,1.65) {$q^3-p^3$};
  \node[lab,anchor=west] at (2.10,1.9) {$q^1-p^1$};

  \node[region] at (-2.20,1.24) {$A$};
  \node[region] at (-2.35,-0.12) {$B$};
  \node[region] at (-0.38,-1.35) {$C$};
  \node[region] at (1.40,-1.05) {$D$};
  \node[region] at (1.72,-0.06) {$E$};
\end{tikzpicture}
\caption{Lotteries in Axiom \ref{ax:n-SI} with $n=k=|Z|=3$. The lottery differences all live in a 2-dimensional plane which extends the unit simplex (all vectors on this plane sum up to zero). Utilities are normalized such that they  are also points on this plane.} \label{fig:nSI}
\end{figure}
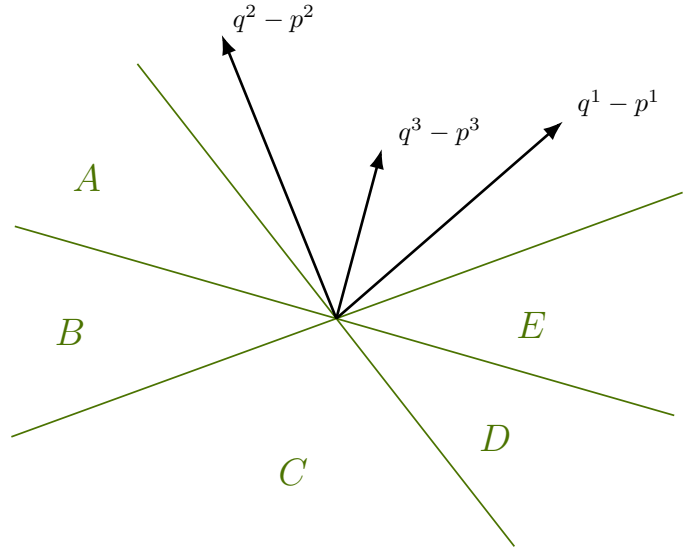

\vspace*{-1em}

\begin{theorem}\label{thm:EU}
 A dataset of size $n$ satisfies $n$-strong independence if and only if it is consistent with EU.	Moreover, if $m:=|Z|$ is such that $m<n$, then it is enough to check $m$-strong independence.
\end{theorem}

As mentioned above, Axioms \ref{ax:ind}, \ref{ax:SI}, and \ref{ax:n-SI}  need to be verified for all enumerations of the dataset. This seems computationally costly (we make this statement precise in Section \ref{sec:checking}). Consider instead,  the following condition, which  needs to be checked only once.

\begin{condition} \label{axiom_disc_basic}
    $0 \not \in \conv \{ (2 y_i -1) (p_i - q_i)\}_{i=1}^n$.
\end{condition}

We derive our Theorem \ref{thm:EU} as a corollary of the following result.

\begin{theorem}[\citealt{Fishburn75}]\label{thm:Fishburn1}
    A dataset of size $n$ satisfies Condition \ref{axiom_disc_basic} if and only if it is consistent with EU.	
\end{theorem}

While using axioms is more costly (because of the enumerations), it also has benefits: it is easy to compare Axiom \ref{ax:SI} to Axiom \ref{ax:ind}, while comparing Condition \ref{axiom_disc_basic} to Axiom \ref{ax:ind} seems harder.\footnote{\citet{gilboa2019axiomatizations} argue that axioms are valuable, in part, because they can make a decision rule normatively compelling by expressing it in terms of principles that are easier to understand and endorse. In the case of EU theory, Axiom \ref{ax:SI} serves this role better than Condition \ref{axiom_disc_basic}, as it states the relevant conditions in a more intuitive form.}
 In Section \ref{sec:logit} we show that this comparability has advantages because it helps us reason about experimental confounds.  In the next section, we show how to use this machinery to study new questions in stochastic choice.


\section{Stochastic Choice}\label{sec:REU}
%
%
%

%

Now a \emph{dataset}  is  $D=\left((p_i, q_i, \varrho_i)\right)_{i=1}^n,$ where $p_i\neq q_i$ are both in $\Delta(Z)$, and $\varrho_i\in [0, 1]$ is the fraction of subjects who chose lottery $p_i$ over $q_i$. Recall that  $D$ is \emph{REU-consistent} if there exists a Borel probability distribution  $\mu\in \Delta(\bR^Z)$ such that $\varrho_i=\mu(\{u\in \bR^Z: \langle u, p_i-q_i \rangle>0\})$ for all $i=1, \ldots, n$, and ties have $\mu$-zero probability.

\cites{GP06} axioms are necessary and sufficient for REU if the analyst observes choices from every finite menu of lotteries. Their  main axiom restricted to size-2 datasets says the following.\footnote{The convention here that the superscripts are different enumerations of the dataset, such as $(p^1, q^1, \varrho^1)=(p_2, q_2, \varrho_2)$ and $(p^2, q^2, \varrho^2)=(q_1, p_1, 1-\varrho_1)$.}

\begin{axiom}[Linearity]\label{ax:lin}
		For all $\alpha\in (0,1)$ and $r\in \Delta(Z)$, if 
		\begin{align*}
			p^2&=\alpha p^1 + (1-\alpha) r\\
			q^2&=\alpha q^1 + (1-\alpha) r,
		\end{align*}
		then $\varrho^1=\varrho^2$.
\end{axiom}

For example, \cites{KT79} certainty effect, as originally stated in the paper, is a violation of linearity: the authors report that $\varrho^1=20\%$ and $\varrho^2=65\%$. 

However, Axiom \ref{ax:lin} does not capture the whole strength of REU on finite datasets. We need to strengthen it to the following condition.

\begin{axiom}[Strong Linearity]\label{ax:slin}
For all $\lambda>0$, if $p^2-q^2=\lambda (p^1-q^1)$, then $\varrho^1=\varrho^2$.
\end{axiom}

At the level of the axioms, strong linearity is very similar to strong independence, but at the level of the representations, REU is a more complicated object than  EU  (a distribution over utility functions vs. a utility function). Nevertheless, as our Theorem \ref{thm:REU} shows,  strong linearity is equivalent to REU-consistency for datasets of size-2.

For datasets of size-3, we need the following axiom. 

\begin{axiom}[Triangulation]\label{ax:triang}
If $p^3-q^3=\lambda^1 (p^1-q^1) + \lambda^2 (p^2-q^2) $ for $\lambda^1, \lambda^2\geq 0$, then 
$$\varrho^1+\varrho^2 -1 \leq \varrho^3 \leq \varrho^1+\varrho^2.$$
\end{axiom}


Necessity of Triangulation can be seen in Figure \ref{fig:nSI}, which we previously used to illustrate necessity of $3$-Strong Independence. Under REU, we have 
\begin{align*}
	\varrho^1&=\mu(A) + \mu(B) + \mu(C) \\
	\varrho^2&=\mu(C) + \mu(D) + \mu(E)\\	
	\varrho^3&=\mu(B) + \mu(C) +\mu(D).
\end{align*}
Because $\mu$ is a probability measure, it follows immediately that
$$\varrho^1+\varrho^2-1 \leq \varrho^3\leq \varrho^1+\varrho^2.$$
As our Theorem \ref{thm:REU} shows, a dataset of size 3  is REU-consistent if and only if it satisfies Axioms \ref{ax:slin} and \ref{ax:triang}. One would hope that if we generalize  triangulation, we can deal with bigger datasets. We will show that the following axiom is necessary.

\begin{axiom}[$k$-Triangulation]\label{ax:n-triang}
    If $p^k - q^k = \sum_{i=1}^{k-1} \lambda_i (p^i - q^i)$ for $\lambda_i \geq 0$, then
    \begin{align}
        \sum_{i=1}^{k-1} \varrho( p^i, q^i) - (k-2) \leq \varrho(p^k, q^k) \leq \sum_{i=1}^{k-1} \varrho(p^i, q^i). \nonumber 
     \end{align}
\end{axiom}

Note that $n$-strong independence implies $k$-strong independence for $k\leq n$ (by setting the corresponding $n-k$ values of $\lambda$ to zero). But this is no longer true for triangulation. For example it is easy to see that $3$-triangulation does not imply $2$-triangulation.

Our next result shows that checking $k$-triangulation for all $k\leq n$ is sufficient for REU, but only as long as $n\leq 4$. 

\begin{theorem}\label{thm:REU}
    Suppose $D$ is a size-$n$ dataset. If $D$ is REU-consistent, then it satisfies $k$-triangulation for all $k \in \{2,...,n\}$. The converse is true for $n\leq 4$, but not for $n>4$. 
\end{theorem}

Because $n$-triangulation is not strong enough for $n>4$, we need to take a different approach that works directly with the deterministic preferences consistent with EU (in fact, our proof of Theorem \ref{thm:REU} also relies on the following construction). This approach has been taken by \cite{kitamura2018nonparametric} and \cite{smeulders2021nonparametric} in the context of consumer choice. Take the $n$ lottery pairs in the dataset $\left((p_i, q_i)\right)_{i=1}^n$. Fix $u\in \bR^Z$ such that $\langle u, p_i-q_i\rangle \neq 0$ for all $i$ (recall that ties have $\mu$-zero probability). Now imagine a vector of dimension $n$ that equals 1 if $\langle u, p_i-q_i\rangle>0$ and zero otherwise. Let $Y_n\subseteq\{0, 1\}^n$ denote this set of EU-consistent choice vectors. By Theorem \ref{thm:EU}, the set $Y_n$ consists of all vectors $y\in \{0, 1\}^n$ such that  Axiom \ref{ax:n-SI} holds on dataset $\left( (p_i, q_i, y_i)\right)_{i=1}^n$. Define the vector $\varrho:=(\varrho_1, \ldots, \varrho_n)$. 

\begin{proposition}\label{prop:1}
    $D$ is REU-consistent iff $\varrho\in \textrm{conv}(Y_n)$. 
\end{proposition}

Thus, checking  REU consistency is equivalent to checking a convex hull condition; in fact, \cite{kitamura2018nonparametric} show that this can be weakened to checking a convex cone condition. We discuss the computational aspects of these procedures in Section \ref{sec:checking}. In Appendix \ref{app:ARSP} we discuss the Axiom of Revealed Stochastic Preference (ARSP) of \cite{McFaddenRichter}, which is derived from the convex hull condition.

As always in stochastic choice theory, axioms are meant to be imposed on limiting choice frequencies. If we have only a few repetitions of each question, an axiom can be rejected even if the ground truth is REU because of the randomness of sampling (and it can be accepted if the ground truth is not). We will not be focusing on ``standard errors'' for axioms in this paper. A bootstrap approach has been developed by \cite{kitamura2018nonparametric}.   A Bayesian approach for testing REU has been developed by \cite{mccausland2020testing}.

\section{Applications}\label{sec:applications}

\subsection{Experiments with random lottery pairs}\label{sec:13k}

A growing literature ties the hands of the experimentalist by drawing at random the lottery pairs that are offered to the subjects in the experiment. For example, \cite{Erevetal17}  put a finite grid on the probability values and draw prizes from a continuous distribution over $Z=\bR$. The \texttt{choices13k} dataset \citep{13k} follows this protocol.

It has been well understood that asking random questions instead of carefully tailoring the experiment has a price because the resulting test loses statistical power. As we make it formal below, the power goes all the way down to zero. The resulting dataset will satisfy EU with probability one, \emph{regardless of the answers given by the subjects}! This result pertains not only to \texttt{choices13k} and similar datasets, but also to any machine-learning models trained on this data: if the training data cannot distinguish EU from non-EU, how could the model?

\begin{corollary}\label{cor:random-erev}
	Suppose that each lottery in the experiment is drawn at random: its support size  according to some discrete distribution, its outcomes according to a continuous distribution over $\mathbb{R}$, and its probability values according to an arbitrary distribution. Then with probability one, the dataset is EU-consistent, regardless the actual choice pattern (the values of $y_i$), and the number of lottery pairs  $n$. 
\end{corollary}


Intuitively, this result follows because the premise of Axiom \ref{ax:n-SI} requires the $n$ vector of lottery differences to be linearly dependent. With a fixed outcome set $Z$, this would hold for a high enough $n$ because the space $\bR^Z$ would get ``saturated.'' But when prizes are drawn from a continuous distribution, the dimensionality of the space grows with $n$. 

Note that in any real experiment, the probability values and  outcomes are rounded (perhaps at the second decimal). In this case, the space of lottery differences could get saturated, such that the  probability of drawing a dataset that violates Axiom \ref{ax:n-SI} is not exactly zero. In practice, however, this is not a concern: as we report in the first row of Table \ref{tab:zeroprob}, the number of violations of EU equals zero in the \texttt{choices13k} dataset \citep{13k}.

\begin{table}[h!]
\centering
\begin{tabular}{l|cc|cc}
\toprule
  & \multicolumn{2}{c|}{deterministic preference} & \multicolumn{2}{c}{stochastic preference} \\
\hline
Violations of ...      & actual & max   & actual & max   \\
\hline
EU                     & 0      & 0     & 0      & 0     \\
Monotone EU            & 92     & 1,024 & 414    & 2,035 \\
Monotone \& concave EU & 616    & 2,008 & 1,791  & 4,335 \\
\hline
$N$ & \multicolumn{2}{c}{2,183} & \multicolumn{2}{c}{4,741} \\
\bottomrule
\end{tabular}
\caption{For each subject, the dataset consists of all the questions she answered. $N$ is the number of subjects that have answered at least two questions. The  ``deterministic preference''  column restricts attention to questions that subjects have had consistent answers to across all 5 trials. The ``stochastic preference'' column defines a subject's deterministic choice as the option chosen more often.  The column ``actual'' is the actual number of violations in the data, and ``max'' is the maximum possible number of violations. In Supplementary Appendix \ref{app:13k} we discuss the details of this exercise and present additional  results.}
\label{tab:zeroprob}
\end{table}

The second and third rows of Table \ref{tab:zeroprob} report the number of violations of EU with the additional restrictions that $u$ is increasing and concave. With these additional restrictions (discussed in detail in Appendix \ref{app:monotandconcave}), EU does have empirical bite even in datasets with randomly generated lottery pairs. For example, 2,035 of the 4,741 subjects faced lottery pairs satisfying the premise of our monotone-EU axiom, and 414 of these 2,035 subjects actually violated EU with monotone utility. However, note that the motivating anomalies, such as the Allais paradox and certainty effect, are violations of pure EU (without any added restrictions), so they cannot be detected using datasets with random lottery pairs.

We remark that in  the \texttt{choices13k} dataset, each subject faced the same lottery pair 5 times. Taken literally, this violates the continuity assumption of Corollary \ref{cor:random-erev} and makes it possible to violate EU by choosing different answers on different repetitions of same the question. In Table \ref{tab:zeroprob}, we address this by (a) restricting the sample to questions that subjects have
had consistent answers to across all 5 trials (deterministic preference), or (b) looking at the whole sample but defining a subject’s deterministic choice as the option chosen more often (defining a preference relation  this way has been a common approach at least since \cite{Tversky69}; we call this relation the stochastic preference). Another approach would be to define $\varrho_i$ as the frequency with which the subject chose the first option, then test REU-consistency on the stochastic dataset $D^*=\left((p_i, q_i, \varrho_i)\right)_{i=1}^n$. This leads to the same conclusion as above.

\begin{corollary}\label{cor:2}
    Suppose that each lottery in a stochastic dataset $D^*$ is drawn at random as in Corollary \ref{cor:random-erev}. Then with probability one, $D^*$ is REU-consistent.
\end{corollary}

\subsection{Spurious anomalies}\label{sec:logit}
Experimental economists often use a model of choice that is EU with i.i.d.\ noise. In this model, there is a fixed, deterministic Bernoulli utility $u\in \bR^Z$, and the probability of choosing the lottery $p^i$ over $q^i$ equals $\varrho^i=\text{Prob}(\langle u, p^i \rangle +\epsilon_p^i > \langle u, q^i \rangle +\epsilon_q^i)$, where $\epsilon_p^i, \epsilon_q^i$ are i.i.d.\ random variables.

It is well known that such models can spuriously generate violations of Axiom \ref{ax:lin} \citep{ballinger1997decisions,Loomes05,hey2005we,blavatskyy2007stochastic,Wilcox08,blavatskyy2010reverse,bhatia2017noisy}. Intuitively, this happens because if $p^2=\alpha p^1+ (1-\alpha) r$ and $q^2=\alpha q^1+ (1-\alpha)  r$, then the expected utility difference between $p^2$ and $q^2$ is smaller (by a factor of $\alpha$) than the expected utility difference between $p^1$ and $q^1$, so it is more easily swamped by the $\epsilon$ noise.

This is a serious confound because we cannot tell whether the observed violations reflect “true” non-EU preferences or arise mechanically, as in an EU-plus-noise model.  \cite{Mcgranaghan24} proposed a solution to this problem based on eliciting certainty equivalents of lotteries. We propose an alternative solution: even when looking only at pairwise lottery choices, we design tests that are satisfied if the ground truth is EU plus noise, but fail if the ground truth is non-EU.

To see how this can be done, consider the special case of our strong linearity axiom with $\lambda=1$, i.e., a size-2 dataset where $p^1-q^1=p^2-q^2$. In this case, because expected utility is linear in probabilities, the expected utility differences are the same
\begin{align*}
\langle u, p^1-q^1\rangle &=\langle u, p^2-q^2\rangle\\[-1em] 
\shortintertext{ and therefore }
\text{Prob}(\langle u, p^1 \rangle +\epsilon_p^1 > \langle u, q^1 \rangle +\epsilon_q^1)&=\text{Prob}(\langle u, p^2 \rangle +\epsilon_p^2 > \langle u, q^2 \rangle +\epsilon_q^2),
\end{align*}
or equivalently, $\varrho^1=\varrho^2$. Therefore, a test based on strong linearity with $\lambda=1$  is passed if the ground truth is EU plus i.i.d.\ noise (and more generally for any Fechnerian EU ground truth).  

\begin{example}
Let  $Z=\{0, 3, 4\}$ and consider the lotteries presented in Table \ref{tab:lambda} and Figure \ref{fig:new-lambda}.



\newlength{\panelheight}
\setlength{\panelheight}{0.3\textheight} 

\begin{figure}[!h]
    \centering

    \begin{minipage}[t]{0.48\textwidth}
        \centering
        \begin{minipage}[c][\panelheight][c]{\linewidth}
            \centering
            \begin{tabular}{cccccc}
                \toprule
                 & $p^1$ & $\hat p^1$ & $q^1$ & $p^2$ & $q^2$\\
                \midrule
                $z=0$ & .20 & .05 & 0 & .80 & .75\\
                $z=3$ & 0   & .75 & 1 & 0   & .25\\
                $z=4$ & .80 & .20 & 0 & .20 & 0\\
                \bottomrule
            \end{tabular}
        \end{minipage}

        \captionof{table}{Lotteries in the certainty effect and the modified certainty effect.}
        \label{tab:lambda}
    \end{minipage}
    \hfill
    \begin{minipage}[t]{0.45\textwidth}
        \centering
        \begin{minipage}[c][\panelheight][c]{\linewidth}
            \centering
\begin{tikzpicture}[scale=1.5]
  \tikzset{
    simplex/.style={draw=black,line width=0.9pt},
    redline/.style={draw=red,line width=0.45pt},
    greenpt/.style={circle,fill=green!35!black,inner sep=1.6pt},
    lab/.style={font=\footnotesize,inner sep=1pt}
  }

  \coordinate (L) at (0,0);
  \coordinate (R) at (4.25,0);
  \coordinate (T) at (2.125,3.10);

  \coordinate (p1)    at (0.85,0);        
  \coordinate (q1)    at (T);             
  \coordinate (phat1) at (1.80625,2.325); 
  \coordinate (p2)    at (3.40,0);        
  \coordinate (q2)    at (3.71875,0.775); 

  \draw[simplex] (L)--(T)--(R)--cycle;

  \draw[redline] (p1)--(q1);
  \draw[redline] (p2)--(q2);

  \foreach \p in {p1,q1,phat1,p2,q2}
    \node[greenpt] at (\p) {};

  \node[lab,anchor=east] at ($(L)+(-0.16,-0.02)$) {$4$};
  \node[lab,anchor=south] at ($(T)+(0,0.15)$) {$3$};
  \node[lab,anchor=west] at ($(R)+(0.18,-0.02)$) {$0$};

  \node[lab,black!35!black,anchor=north] at ($(p1)+(0,-0.15)$) {$p^1$};
  \node[lab,black!35!black,anchor=west] at ($(phat1)+(0.12,0.00)$) {$\hat p^1$};
  \node[lab,black!35!black,anchor=west] at ($(q1)+(0.18,0.00)$) {$q^1$};
  \node[lab,black!35!black,anchor=north] at ($(p2)+(0,-0.15)$) {$p^2$};
  \node[lab,black!35!black,anchor=west] at ($(q2)+(0.16,0.02)$) {$q^2$};
\end{tikzpicture}

        \end{minipage}

        \caption{Lotteries in the certainty effect and the modified certainty effect.}
        \label{fig:new-lambda}
    \end{minipage}

\end{figure}

The lottery pairs $\{(p^1, q^1), (p^2, q^2)\}$ correspond to \cites{KT79} original certainty effect. That experiment is designed to test Axiom \ref{ax:lin} (with $\alpha=.25$). We propose to modify the lottery pairs to $\{(\hat p^1, q^1), (p^2, q^2)\}$ and test Axiom \ref{ax:slin} (with $\lambda=1$). Assuming  Bernoulli utility $u(z) = z$, we see that
\begin{align*}
    \langle u, \hat p^1 - q^1 \rangle &  = 3.05 - 3 = 0.05 \\
    \langle u,  p^1 - q^1 \rangle & = 3.20 - 3 = 0.20 \\
    \langle u,  p^2 - q^2 \rangle & = 0.80 - 0.75 = 0.05.
\end{align*}
Under i.i.d. noise, the test of  Axiom \ref{ax:lin} on  $\{(p^1, q^1), (p^2, q^2)\}$ fails, while the test of Axiom \ref{ax:slin} on $\{(\hat p^1, q^1), (p^2, q^2)\}$ succeeds. 

Thus, the test passes if the ground truth is EU plus noise.  Conversely, if the test fails, then the ground truth cannot be EU plus noise (nor can it be REU, by Theorem \ref{thm:REU}).  This does not mean, however, that the test fails for every non-EU ground truth: even non-EU preferences may happen to agree with EU on  $\{(\hat p^1, q^1), (p^2, q^2)\}$. Nevertheless, as we show in  Supplementary Appendix \ref{app:fail-test}, the test does fail for typical parameterizations of non-EU, and repeating the test with perturbed lottery values can ensure that it fails for all parameters. \finex
\end{example}

\subsection{Automatic Generation of Anomalies}\label{sec:MR}

In this section, we show how our insights can be used for the automatic generation of anomalies. Pioneered by \cite{MR}, the general idea works as follows. Take a \emph{ground truth} utility function (outside of the EU class) and generate choices for a given  $2n$-tuple of lotteries to form a dataset $D$. Check if $D$ is EU-consistent. If not, then $D$ contains an anomaly. Ideally, if we search thoroughly across the space of all $2n$-tuples of lotteries and the space of all possible ground truths, we could uncover all possible anomalies. 

To implement this idea practically, \cite{MR} (henceforth MR),  take a parametric specification of non-cumulative prospect theory,  where individuals evaluate lotteries using the probability weighting function  $\pi_z(p ; \delta, \gamma)=\frac{\delta p(z)^\gamma}{\delta p(z)^\gamma+\sum_{z' \neq z} p(z')^\gamma}$, and calibrate parameter values $\theta = (\delta, \gamma)$ as $\theta_1=(0.726,0.309)$, $\theta_2=(0.926,0.377)$, and $\theta_3=(1.063,0.451)$ following  \cite{lattimore1992influence}.\footnote{MR test for EU with monotone utility, so it is important that the ground truths allow for violations of FOSD-monotonicity, as is the case with non-cumulative prospect theory.}

Given a ground truth $\theta$ and a $2n$-tuple of lotteries $\tau$, MR simulate the stochastic choice function $\varrho_\theta(\tau)$ by adding logit noise to each computed prospect-theory value (similar to EU plus noise from Section \ref{sec:logit}). To check EU-consistency, they search over two parameterized families of monotone Bernoulli utility functions (using splines or polynomials), and again add logit shocks to generate stochastic choice $\varrho_\psi(\tau)$, where $\psi$ is the parametrization of EU.   To search over  lottery tuples, MR use an \emph{adversarial} procedure, $\max_\tau \min_\psi L(\varrho_\psi(\tau), \varrho_\theta(\tau)),$ where $L$ is a loss function. This selects datasets which are most ``anomalous'' according to $L$. MR also consider a different \emph{example morphing} procedure that operates similarly. Both procedures search continuously over the space of prizes and probabilities.

MR classify the most ``anomalous'' datasets (of size $n=2$) into three categories: Dominated Consequence Effect, Reverse Dominated Consequence Effect, and Strict Dominance Effect (we discuss the definitions of these categories in Appendix \ref{app:monotandconcave}). MR report  that 96-98\% of anomalies fall into these categories.

However, we now show that within the domain of deterministic choices without logit shocks, there are many anomalies that escape the MR classification for the same ground truths. It is important to look at the deterministic case because prospect theory and EU are inherently theories of deterministic choice. Moreover, adding randomness generates concerns about spurious anomalies that we discussed in Section \ref{sec:logit}, so looking at the deterministic case should give us a cleaner view.

Specifically, for each ground truth $\theta$ and a given $2n$-tuple of lotteries, we compute the vector $y$ of deterministic choices predicted by the ground truth utility function $U_{(\delta, \gamma)} (p) = \sum_z z \,  \pi_z(p; \delta, \gamma)$. This yields  a dataset $D$. We can now check if $D$ is EU-consistent (or, more precisely, consistent with monotone-EU), as described in Section \ref{sec:checking}. We then do a grid search over the space of $2n$-tuples of lotteries.

Table \ref{tab:MR_combined} presents our results and we will describe it in parts. The first three columns correspond to the original MR ground truths, and the first three rows correspond to the MR categories; the last row is the total number of anomalies. The fact that none of them belongs to the MR classes may suggest that our deterministic procedure is too coarse. The next four columns of Table  \ref{tab:MR_combined}  shows that this is not the case. These columns correspond to ground truth parameterizations that optimize the number of generated anomalies, and we see that the MR anomalies do appear, but  only consist of about 10\% of all the anomalies. 

It is important to note that MR operate in the context of EU with \emph{monotone} Bernoulli utility functions. On the full domain, this class is characterized by the vNM axioms plus FOSD-monotonicity (Axiom \ref{ax:dom}). However, as seen in Table \ref{tab:MR_combined}, this axiom and strong independence (Axiom \ref{ax:SI}) are rarely or never violated. On the finite domain, all counterexamples are caught by Axiom \ref{ax:m-2-cancellation}, which, as we discuss in Appendix \ref{app:monotandconcave}, is equivalent to EU with monotone utility on finite datasets. In Appendix \ref{app:monotandconcave} we also discuss the details of the MR classification and propose a broader categorization to capture  ``uncategorized'' anomalies.

\begin{table}[h!]
\centering
\begin{tabular}{l|ccc|cccc}
\toprule
 & \multicolumn{3}{c|}{MR parameters} & \multicolumn{4}{c}{optimized parameters} \\
 & $\theta_1$ & $\theta_2$ & $\theta_3$ & $\theta_4$ & $\theta_5$ & $\theta_6$ & $\theta_7$ \\
\hline
Dominated Consequence              & 0    & 0    & 0    & 258  & 235  & 172  & 0    \\
Reverse Dominated Consequence      & 0    & 0    & 0    & 268  & 257  & 211  & 0    \\
Strict Dominance                   & 0    & 0    & 0    & 250  & 227  & 166  & 0    \\
\hline
Violations of Axiom \ref{ax:SI}                & 0    & 0    & 0    & 0    & 0    & 0    & 0    \\
Violations of Axiom \ref{ax:dom}              & 0    & 0    & 0    & 82   & 98   & 93   & 0    \\
Violations of Axiom \ref{ax:m-2-cancellation}  & 834  & 521  & 293  & 2391 & 2702 & 2488 & 2098 \\
\hline
Total anomalies                   & 834  & 521  & 293  & 2391 & 2702 & 2488 & 2098 \\
\bottomrule
\end{tabular}
\caption{Number of anomalies out of 749{,}700 size-2 datasets with MR ground truths (parameters $\theta_1$--$\theta_3$) and optimized ground truths ($\theta_4$--$\theta_7$), using random gridding of probability values (described in Supplementary Appendix \ref{app:MRanomaliestwo}).}
\label{tab:MR_combined}
\end{table}

The main difference between MR's approach and ours is the addition of logit shocks, which mechanically creates  a spectrum of ``anomalousness,'' whereas for us $D$ either is or is not an anomaly. By looking at only the ``most anomalous'' datasets, MR looks at a biased sample. Another difference is the search algorithm over lottery tuples, but this perhaps plays less of a role because MR's optimization problem is nonconvex, so ultimately they rely on random seeding, which is similar to our grid search. Yet another difference is that we are checking EU exactly, while MR resort to parametric approximations; we do not know to what extent this plays a role.

\subsection{Computational aspects of testing}\label{sec:checking}

In this section, we show that it is not that hard to test expected utility directly: it boils down to solving a linear program.  In fact, it is easier to test EU directly, rather than testing its axioms.   By definition, a dataset $D$ is consistent with EU iff 
\begin{equation}\label{eq:EU-repr}
	\exists_{u\in \bR^Z} \forall_{i=1, \ldots, n} (2y_i-1)\langle u, p_i-q_i\rangle >0.
\end{equation}

For any pair of lotteries $p, q$, denote the half-space 
$$H^{p-q}:=\{u\in \bR^Z: \langle u, p-q\rangle > 0\}.$$
This is the set of Bernoulli utility functions that rationalize the choice of $p$ over $q$. Notice that $D$ is consistent with EU iff the intersection of its induced half-spaces is nonempty
\begin{equation}\label{eq:intersection1}
\bigcap_{i=1}^n H^{(2y_i-1)(p_i-q_i)}\neq \emptyset.	
\end{equation}


Checking if this intersection is nonempty is a linear program and efficient algorithms can solve this problem in polynomial time.\footnote{Since these are strict inequalities, we need to turn each strict inequality to a weak one with a slack variable instead of zero on the right hand side, and then maximize over the slack variable. This is a linear program. Let $n$ be the size of $D$ and $m$ be the total number of prizes used in $D$ (the cardinality of the union of the supports of all lotteries in $D$) and $L$ the number of bits per variable. The \cite{Karmarkar} algorithm, for instance, requires $O(m^{\frac{3}{2}}n^{2}L)$ operations to solve our linear program.}  Alternatively, we can  check Fishburn's Condition \ref{axiom_disc_basic}. This also corresponds to a linear program, with $n$ variables and $m+1$ equality constraints, since $0 \in \conv\{v_i\}_{i=1}^n$ if and only if $\exists \lambda_i \geq 0$ such that $\sum_{i=1}^n \lambda_i v_i =0$ and $\sum_{i=1}^n \lambda_i = 1$.






The same idea extends to monotone EU. Label the outcome space $Z = \{z_1,...,z_m\}$ such that $z_1 < z_2 < ...  < z_m$, and let $e_i$ correspond to the degenerate lottery that places probability 1 on $z_i$. Observe that a dataset $D$ is consistent with monotone EU iff   $\exists_{u\in \bR^Z} \forall_{i=1, \ldots, n} (2y_i-1)\langle u, p_i-q_i\rangle >0$ and $\forall_{i>j} \langle u, e_i - e_j\rangle > 0$.   In terms of half spaces, the condition is then
\begin{align*}
    \left( \bigcap_{i=1}^n H^{(2y_i-1)(p_i-q_i)} \right) \cap \left( \bigcap_{i > j} H^{ e_i - e_j} \right) \neq \emptyset,
\end{align*}
which is the natural analog of \eqref{eq:intersection1}.\footnote{Similarly, the analog to Condition \ref{axiom_disc_basic} is $0 \not \in \conv(\{(2y_i - 1)(p_i -q_i)\}_{i=1}^n \cup \{e_i - e_j\}_{i > j})$, which adds $\binom{m}{2}$ more variables to the linear program that checks Condition \ref{axiom_disc_basic}.} Hence, we can also verify monotone EU directly in polynomial time.


As a side note, we consider how costly it is to verify whether our axioms hold on $D$. Theorems \ref{thm:EU} and \ref{thm:REU} require us to check the axioms for all enumerations of $D$. For this reason, testing the axioms is actually computationally harder than directly testing the theories. For instance, to check $n$-strong independence, we would need to check all $2^n n!$ enumerations.  Of course, we are not advocating for checking our axioms in this manner. The axioms were formulated this way so that they are comparable with the classic axioms of decision theory, while helping us  reason about the properties of various models, such as  those discussed in Section \ref{sec:logit}.

The case of stochastic datasets is more complicated. Proposition \ref{prop:1} says that $D$ is REU-consistent if $\varrho\in \conv(Y_n)$, where $Y_n$ is the set of all EU-consistent deterministic choices between the lottery pairs of $D$. Given the set $Y_n$, this condition can be verified by linear programming. However, computing the set $Y_n$ itself is costly because this set can be massive. For example, if the lottery pairs are generated as in Section \ref{sec:13k}, then with probability one the set $Y_n$ has $2^n$ elements (of course computing the convex hull is easy in this case). 


By contrast, a researcher with control over the experimental design can choose the questions so that the set $Y_n$ is known by construction. For instance, if we restrict attention to binary menus  $\{p_i, q_i\}$ such that $p_i - q_i$ and $p_j - q_j$ lie on the same ray for all $i,j$, then $2$-Strong Independence implies that $Y_n = \{ (1,...,1), (0,...,0)\}$. 



A similar issue arises in the literature on abstract choice. Suppose $\tilde Y_n \supseteq Y_n$ is the collection of vectors $y \in \{0,1\}^n$ such that the dataset $D= ((p_i, q_i, y_i))_{i=1}^n$ corresponds to a  linear order.  \cite{kitamura2018nonparametric} provides a non-parametric test of the condition $\varrho \in \conv(\tilde Y_n)$. Assume, for simplicity, that all $p_i,q_i$ are distinct, in which case $|\tilde Y_n| = 2^n$. The bulk of the computational burden in \cites{kitamura2018nonparametric} procedure originates from how $\tilde Y_n$ scales with $n$, and alternative procedures that are more  computationally tractable   have been proposed \citep{smeulders2021nonparametric, turansick2025alternativeapproachnonparametricanalysis}. In the context of Proposition \ref{prop:1}, we can then interpret Theorem \ref{thm:EU}  as illustrating the \emph{potential} computational gains that arise from imposing the additional structure of  REU.\footnote{When the premise of $m$-Strong Independence holds for various subsets of $D$ of size $m \leq n$, we have $|Y_n| < |\tilde Y_n|$. By contrast, if the premise of $n$-Strong Independence fails on $D$, then $|Y_n| = |\tilde Y_n|$.
}

%

\section{Conclusion}

We extended \cites{Fishburn75} approach to axiomatize Random Expected Utility on finite datasets. We  developed a number of implications for ongoing experimental work, including datasets with random lottery pairs, spurious anomalies generated by logit shocks, and automatic generation of anomalies. We also discussed the computational aspects of testing.

Our approach relies on linearity, so directly extending it to testing non-EU theories will likely be difficult (with some exceptions, as we show in Appendix \ref{app:Yaari}). However, since EU and REU are the baseline theories of our profession, understanding how to test them properly is important, and we hope that our work contributes to this understanding. 

We conclude the paper with a series of appendices. In Appendix \ref{app:monotandconcave}, we discuss EU-consistency with monotone and/or concave Bernoulli utility functions. We also expand on MR's classification of anomalies. In Appendix \ref{app:Yaari}, we present a finite-data axiomatization of \cites{Yaari87} version of cumulative prospect theory. Appendix \ref{app:ARSP} discusses the Axiom of Revealed Stochastic Preference. Appendix \ref{app:EU} contains the proofs in the deterministic case. Appendix \ref{app:REU} contains the proofs in the stochastic case. Additional details for tables and computational procedures are in the Supplemental Appendix (Appendix \ref{app:supp})


\appendix

\section*{Appendices}

\section{Monotone and Concave Utility}\label{app:monotandconcave}
In Section \ref{app:monotandconcave-axioms}, we discuss axioms that are necessary and sufficient for EU with additional restrictions. In Section \ref{app:monotandconcave-MR}, we show how to coarsen the MR classification to capture more anomalies (by our count, we get up from less than 10\% to around 98\% or 100\% for some ground truths).

\subsection{Necessary and Sufficient Axioms}\label{app:monotandconcave-axioms}

Let $Z \subset \bR$. In this section, we show what happens if we add the requirement that the Bernoulli utility function $u$ be monotone or concave. If we observe the entire preference relation over lotteries, this can be axiomatized by adding the following axiom to the vNM axioms. We write $p\fosd q$ if $p$ first-order stochastically dominates $q$ and $p\Fosd q$ if $p\fosd q$ and $p\neq q$.

\begin{axiom}[FOSD-monotonicity]\label{ax:dom}
	If $p\Fosd q$, then $p\succ q$.
\end{axiom}

However, as in the pure-EU case, there are anomalies that satisfy all these axioms.

\begin{example}\label{ex:blackboard}
	Suppose that $p^1\succ q^1$ and $q^2\succ p^2$ where 
	\begin{align*}
	p^2&=\alpha p^1 + (1-\alpha) r\\
	\alpha q^1 + (1-\alpha) r&\Fosd q^2.
	\end{align*}

As in the Allais Paradox, if the analyst only observes the choices between $p^1$ and $q^1$ and between $p^2$ and $q^2$, independence is satisfied (vacuously) and FOSD-monotonicity as well (as long as $p^2 \not \geq_{FOSD} q^2$), yet a monotone EU representation does not exist. If it did, then by independence we must have $p^2\succ \alpha q^1 + (1-\alpha) r$ and by FOSD-monotonicity $\alpha q^1 + (1-\alpha) r \succ q^2$, so by transitivity $p^2\succ q^2$. The key is that choice comparisons with lottery $\alpha q^1 + (1-\alpha) r$ are not contained in the dataset.\finex
\end{example}

To rule out such examples, we need an axiom that combines independence and FOSD-monotonicity. Consider the following one.

\begin{axiom}[Cancellation]\label{ax:m-2-cancellation}
	For all $\alpha\in [0, 1]$, if $p^1\succ q^1$ and 
	$$\alpha q^1+ (1-\alpha) p^2 \geq_{FOSD} \alpha p^1+ (1-\alpha) q^2$$
	then $p^2\succ q^2$.
\end{axiom}

If $p^1 \succ q^1$, mixing any lottery with $q^1$ is ``worse" than doing so with $p^1$. Cancellation says that if $p^2$ is so much ``better" than $q^2$ such that the mixture $\alpha q^1+ (1-\alpha) p^2$ first-order stochastically dominates $\alpha p^1+ (1-\alpha) q^2$ (i.e., offsetting the initial disadvantage of $q^1$ relative to $p^1$), then  it must be that $p^2$ is strictly preferred to $q^2$. 

Cancellation extends to size-$n$ datasets as follows. 

\begin{axiom}[$n$-Cancellation]\label{ax:m-cancellation}
	For all $\alpha \in \Delta(n)$, if $p^i\succ q^i$ for $i =1, \ldots, n-1$ and 
	$$\sum_{i=1}^{n-1} \alpha^i q^i + \alpha^n p^n \geq_{FOSD} \sum_{i=1}^{n-1} \alpha^i p^i + \alpha^n q^n,$$
	then $p^n\succ q^n$.
\end{axiom}

Recall that $p\geq_{FOSD} q$ if the expected utility of $p$ is weakly higher than of $q$ for all strictly increasing utility functions. In general, we could imagine an incomplete relation $\geq_{\cU}$ given by a convex set of utility functions $\cU$. For example, second order stochastic dominance is given by the class of all increasing and concave utility functions. The following axiom extends the cancellation idea along these lines.

\begin{axiom}[$\cU$-Cancellation]\label{ax:U-cancellation}
	For all $\alpha \in \Delta(n)$, if $p^i\succ q^i$ for $i =1, \ldots, n-1$ and 
	$$\sum_{i=1}^{n-1} \alpha^i q^i + \alpha^n p^n \geq_{\cU} \sum_{i=1}^{n-1} \alpha^i p^i + \alpha^n q^n,$$
	then $p^n\succ q^n$.
\end{axiom}

The following result is an implication of Corollary 3(a) of \cite{Fishburn75}.

\begin{theorem}\label{thm:EU-2}
A dataset of size $n$	satisfies $\cU$-Cancellation if and only if it is consistent with EU from class $\cU$. 
\end{theorem}

Because cancellation-like axioms rule out all anomalies, they are not very useful for classifying them into various types. For that purpose, we develop a handful of weaker axioms in the following section

\subsection{Axiomatic Classification of Anomalies}\label{app:monotandconcave-MR}
Consider the following axiom, designed precisely to rule out Example \ref{ex:blackboard}.

\begin{axiom}[Right Combination]\label{ax:Left}
	For all $\alpha\in [0, 1]$ and $r\in \Delta(Z)$, if $p^1\succ q^1$ and
	\begin{align*}
	p^2&=\alpha p^1 + (1-\alpha) r\\
	\alpha q^1 + (1-\alpha) r&\Fosd q^2,
	\end{align*}
	then $p^2\succ q^2$.	
\end{axiom}

This axiom combines one application of independence and one application of FOSD-monotonicity.  One could also imagine its mirror image. 

\begin{axiom}[Left Combination]\label{ax:Right}
	For all $\alpha\in [0, 1]$ and $r\in \Delta(Z)$, if $p^1\succ q^1$ and    
	\begin{align*}
	p^2&\Fosd \alpha p^1 + (1-\alpha) r\\
	\alpha q^1 + (1-\alpha) r&= q^2, 	
	\end{align*}
 then $p^2\succ q^2$.		
\end{axiom}

There are also other possibilities, such as applying FOSD-monotonicity twice. 

\begin{axiom}[Double Combination]\label{ax:Middle}
    If $p^1\succ q^1$ and
   	\begin{align*}
   	p^2&\Fosd p^1 \\
   	 q^1 &\Fosd q^2,
   	\end{align*}
   	then $p^2\succ q^2$.
\end{axiom}

MR's three categories of anomalies are special kinds of violations of the above three axioms: Axiom \ref{ax:Left} rules out the Dominated Consequence Effect by setting  $r=\delta_x$ and  $q^2=\beta q^1 + (1-\beta) \delta_y$ for some $\beta<\alpha$ and $y=\min q^1 < x$, ensuring a stark form of FOSD.  Likewise, Axiom \ref{ax:Right} rules out the Reverse Dominated Consequence Effect by setting  $r=\delta_y$ and $p^2=\beta p^1 + (1-\beta) \delta_x$ for some $\beta<\alpha$ and $y<x=\max p^1$. Finally, Axiom \ref{ax:Middle} rules out the Strict Dominance Effect by setting  $p^2=\alpha p^1+ (1-\alpha)\delta_x $  and $q^2=\beta q^1 + (1-\beta) \delta_y$ for some $\alpha,\beta$ and $x=\max p^1, y=\min q^1$.

While our Table \ref{tab:MR_combined} shows that MR's three categories are not very prevalent, one could hope that the generalizations of these categories suggested by the three above axioms be more exhaustive. The following axiom lumps all of three of the above axioms together.

\begin{axiom}[Combination]\label{ax:LMR}
	For all  $\alpha\in [0,1]$ and $r\in \Delta(Z)$, if $p^1\succ q^1$ and 
		\begin{align*}
			p^2\fosd\alpha p^1 &+ (1-\alpha) r\\
			\alpha q^1 &+ (1-\alpha) r\fosd q^2,
		\end{align*}
	then if either one of the $\fosd$ is strict or $\alpha>0$, then $p^2\succ q^2$.
\end{axiom}

As Table \ref{tab:MR_combined_2} shows, even this generalized axiom is far from ruling out all the anomalies. However, the following strengthening of Axiom \ref{ax:LMR} along the lines of strong independence (Axiom \ref{ax:n-SI}) rules out almost all anomalies. 

\begin{axiom}[Strong Combination]\label{ax:strong-LMR}
For all  $\lambda\geq 0$, if $p^1\succ q^1$ and 
	\begin{align*}
		p^2\fosd &\hat p\\
		&\hat q\fosd q^2
	\end{align*}
for $\hat p -\hat q=\lambda (p^1-q^1)$, then if either one of the $\fosd$ is strict or $\lambda>0$, then $p^2\succ q^2$. 	
\end{axiom}

Our exploration in this section illustrates that the key to having a more exhaustive classification of anomalies lies in the same idea behind the way strong independence  (Axiom \ref{ax:SI}) strengthens the independence axiom.

\begin{table}[h!]
\centering
\begin{tabular}{l|ccc|cccc}
\toprule
 & \multicolumn{3}{c|}{MR parameters} & \multicolumn{4}{c}{optimized parameters} \\
 & $\theta_1$ & $\theta_2$ & $\theta_3$ & $\theta_4$ & $\theta_5$ & $\theta_6$ & $\theta_7$ \\
\hline
Dominated Consequence              & 0   & 0   & 0   & 258  & 235  & 172  & 0    \\
Reverse Dominated Consequence      & 0   & 0   & 0   & 268  & 257  & 211  & 0    \\
Strict Dominance                   & 0   & 0   & 0   & 250  & 227  & 166  & 0    \\
\hline
Violations of Axiom \ref{ax:Left}              & 21  & 24  & 14  & 994  & 1069 & 962  & 149  \\
Violations of Axiom \ref{ax:Right}             & 39  & 21  & 6   & 1045 & 1195 & 1064 & 204  \\
Violations of Axiom \ref{ax:Middle}            & 0   & 0   & 0   & 673  & 584  & 423  & 0    \\
Violations of Axiom \ref{ax:LMR}               & 62  & 39  & 18  & 1661 & 1957 & 1846 & 394  \\
Violations of Axiom \ref{ax:strong-LMR}        & 828 & 521 & 293 & 2391 & 2702 & 2488 & 2044 \\
Violations of Axiom \ref{ax:m-2-cancellation}  & 834 & 521 & 293 & 2391 & 2702 & 2488 & 2098 \\
\hline
Total anomalies                                & 834 & 521 & 293 & 2391 & 2702 & 2488 & 2098 \\
\bottomrule
\end{tabular}
\caption{Number of anomalies out of 749{,}700 size-2 datasets with \cite{MR} ground truths (parameters $\theta_1$--$\theta_3$) and optimized ground truths ($\theta_4$--$\theta_7$), using random gridding of probability values.}
\label{tab:MR_combined_2}
\end{table}

\newpage
\section{Probability Weighting}\label{app:Yaari}

Let the prize space $Z$ be some finite subset of $[0,1]$, normalizing the prizes to the unit interval without loss. Write typical elements of $\Delta(Z)$ as $p,q,r$, and their corresponding survival functions (i.e., one minus the  CDF) as $G_p, G_q, G_r$. In our finite support setting, $ G_p(t) = \bP_p(Z > t)  = \sum_{ \{ i : z_i > t \} } p(z_i)$ for any $t \in [0,1]$. Recall that we can write the $p$-expectation of $Z$ using the survival function as
\begin{align*}
    \bE_p[Z] = \int_0^1 G_p(t) \, dt.
\end{align*}

The simplest model of probability weighting, due to \cite{Yaari87}, involves a linear utility function and a nonlinear probability weighting function, such that the utility of lottery $p$ is given by 
\begin{equation}
V_\phi(p):=\int_0^1 \phi(G_p(t)) \, dt     \label{eq:Yaari}
\end{equation}
for some non-decreasing $\phi: [0, 1] \to \bR$. We say that $D$ is Yaari-consistent if there exists a non-decreasing function $\phi:[0, 1]\to \bR$ such that $y_i=1$ iff $V_\phi(p_i) > V_\phi(q_i)$.

Define the inverse of the survival function as 
\begin{align*}
    G_p^{-1}(y) := \inf \bigg\{ t \in [0,1] : G_p(t) \leq y  \bigg\},
\end{align*}
with $G_p^{-1}(1) = 0$. The inverse function $G_p^{-1}: [0,1] \to [0,1]$ is non-increasing and right-continuous, implying that $G_p^{-1}$ is also a survival function with some corresponding lottery that we refer to as $p^{-1}$. Moreover, we have $(p^{-1})^{-1}=p$ for all $p$. 

This allowed \cite{Yaari87} to introduce his dual mixture  $\alpha p \boxplus (1-\alpha) q$ for $\alpha\in [0, 1]$ as the lottery $\left( \alpha p^{-1} + (1-\alpha) q^{-1}\right)^{-1}$. Using this mixture operation, his dual independence axiom is stated as follows.   

\begin{axiom}[Dual Independence]\label{ax:DI}
    For all $\alpha \in [0, 1]$, if $p^1\succ q^1$ and 
    \begin{align*}
        p^2&= \alpha p^1 \boxplus (1-\alpha) r\\    
        q^2&= \alpha q^1 \boxplus (1-\alpha) r,
    \end{align*}
then $p^2\succsim q^2$.         
\end{axiom}

\cite{Yaari87} showed that on the full domain,  dual independence (Axiom \ref{ax:DI}) together with FOSD-monotonicity (Axiom \ref{ax:dom}) and continuity is equivalent to a representation by \eqref{eq:Yaari}. This result holds essentially because the function $V_\phi$ is linear and monotone in the ``dual'' mixures~$\boxplus$. Recall from Appendix \ref{app:monotandconcave} that cancellation (Axiom \ref{ax:m-cancellation}) combines FOSD-monotonicity and independence. Consider now the following axiom, which combines FOSD-monotonicity and dual independence.

\begin{axiom}[Dual Cancellation]\label{ax:DualCancelation}
    If $p^1 \succ q^1$ and
    \begin{align*}
        \alpha q^1 \boxplus  (1-\alpha) p^2 \geq_{FOSD} \alpha p^1 \boxplus (1-\alpha) q^2, \label{eq_neg_cancellation}
    \end{align*}
    then $p^2 \succ q^2$.
\end{axiom}

\begin{theorem}\label{thm:Yaari}
     A size-2 dataset satisfies Axiom \ref{ax:DualCancelation} if and only if it can be rationalized by \eqref{eq:Yaari}.
\end{theorem}

The following result restates the main part of Yaari's original proof. Fix some $\succsim$ and define the induced preference $\succsim^*$ as $p \succsim^* q \iff p^{-1} \succsim q^{-1}$. 

\begin{lemma} \label{lemma_Yaari_tech1}
    The relation $\succsim$ is represented by \eqref{eq:Yaari} if and only if $\succsim^*$ is represented by monotone~EU.
\end{lemma}
\begin{proof}

    \begin{align}
         p \succsim q & \spiff \int_0^1 \phi(G_p(t)) \, dt \geq \int_0^1 \phi(G_q(t)) \, dt \nonumber  \\
         & \spiff - \int_0^1 \phi(y) \, dG_p^{-1}(y) \geq - \int_0^1 \phi(y) \, dG_q^{-1}(y) \tag{by change of variable: $y = G_p(t)$, $y = G_q(t)$} \\
         & \spiff - \int_0^1 \phi(y) \, dG_{p^{-1}}(y) \geq - \int_0^1 \phi(y) \, dG_{q^{-1}}(y)   \tag{by definition of $p^{-1}$} \\
        & \spiff \bE_{p^{-1}}[\phi(Z)] \geq \bE_{q^{-1}}[\phi(Z)] \tag{by definition of EU with a survival function} \\
        & \spiff p^{-1} \succsim^* q^{-1} \nonumber 
    \end{align}

\end{proof}

\begin{proof}[Proof of Theorem \ref{thm:Yaari}]
    Enumerate, without loss, the dataset so that $D = \{(p_1, q_1, 1), (p_2, q_2, 1)\}$. Define the induced dataset $D^* := \{(p_1^{-1}, q_1^{-1}, 1), (p_2^{-1}, q_2^{-1}, 1)\}$. By Lemma \ref{lemma_Yaari_tech1}, we know that $D$ is Yaari-consistent if and only if $D^*$ is rationalizable by monotone EU. By Theorem \ref{thm:EU-2}, we know that the latter statement holds if and only if $D^*$ satisfies Axiom \ref{ax:m-cancellation}. But this is exactly equivalent to $D$ satisfying dual cancellation. 
\end{proof}

The size-$n$ generalization follows analogously.

\section{Axiom of Revealed Stochastic Preference}\label{app:ARSP}

We now show how to adapt \cites{McFaddenRichter} Axiom of Revealed Stochastic Preference (ARSP) to the lottery setting. Intuitively, the original ARSP maximizes over the choices made by all preferences. In Axiom \ref{ax:ARSP-EU}, the maximum is over all EU-preferences (see also \cite{mcfadden2005revealed} for a general treatment with abstract classes of preferences).  While in our previous axioms, the enumerations $\left((p^j, q^j, \varrho^j)\right)$ were sequences of length at most $n=|D|$, ARSP requires us to allow for longer sequences, and for items to be repeated multiple times. Formally, we consider sequences $\left((p^j, q^j, \varrho^j)\right)_{j=1}^k$ such that $(p^j, q^j, \varrho^j)\in D$ or  $(q^j, p^j, 1-\varrho^j)\in D$ for all $j$, and take $i(j)$ as the corresponding index in $D$. Let $\cD$ be the set of all such sequences.

\begin{axiom}[ARSP-EU]\label{ax:ARSP-EU}
For all $\left((p^j, q^j, \varrho^j)\right)_{j=1}^k \in \cD$  
    $$\sum_{j=1}^k \varrho^j \leq \max_{y_*\in Y_*} \sum_{j=1}^k y^j_*.$$
\end{axiom}
Here,  the space $Y_*$ is the extension of the space $Y_n$ to longer sequences of enumerations. Formally, for any $y\in Y_n$ we define $y_*\in \{0, 1\}^k$ by
\begin{equation*}
    y^j_*:=\begin{cases}
        y_i &\text{ if }(p^j, q^j, \varrho^j)=(p_i, q_i, \varrho_i) \text{ for some }i =1, \ldots, n\\
        1-y_i &\text{ if }(p^j, q^j, \varrho^j)=(q_i, p_i, 1-\varrho_i) \text{ for some }i =1, \ldots, n.
    \end{cases}
\end{equation*}
The next proposition then follows by the same logic as in the case of ARSP for random utility.\footnote{For example, the proof of Theorem 1 in \cite{stoye2019revealed} applies directly after specializing Stoye’s collection of rationalizable choice functions to our class of EU-consistent binary relations.}

\begin{proposition}\label{prop:2}
   $D$ is REU-consistent iff it satisfies Axiom \ref{ax:ARSP-EU}.
\end{proposition}

\section{Proofs in the Deterministic Case}\label{app:EU}
We make use of Gordan's lemma, which is a version of the separating hyperplane theorem; see Lemma 1 of \cite{Fishburn75}. 

 \begin{lemma}[Gordan]\label{lem:Gordan}
    Suppose $v_1,...,v_{n} \in \bR^Z$. Exactly one of the following statements hold.
    \begin{enumerate}
        \item There exists some $u \in \bR^Z$ such that $u \cdot v_i > 0$ for all $i \in \{1,...,n\}$. 
        \item  $\textbf{0} \in \conv\{v_1, ..., v_n\}$  
    \end{enumerate}
\end{lemma}

\subsection{Proof of  Theorem \ref{thm:Fishburn1}}

Let the dataset be $D \equiv \{(p_i, q_i, y_i)\}_{i=1}^n$, where $y_i = 1$ if $p_i \succ q_i$ and $y_i = 0$ if $p_i \prec q_i$. It follows directly from Lemma \ref{lem:Gordan} that the Condition \ref{axiom_disc_basic} is necessary and sufficient for $D$ to be EU-consistent. \qed

\subsection{Proof of the first assertion of Theorem \ref{thm:EU}}

We can make Condition \ref{axiom_disc_basic}  more ``axiom-like" by rewriting it as  follows. We superscript the lotteries to emphasize the fact that the axiom must hold for any labeling of the entries in the dataset.

\begin{axiom} \label{axiom_disc_mod}
    Suppose $p^1 \succ q^1$,..., $p^{n-1} \succ q^{n-1}$. If $\textbf{0} \in \conv\{p^1 - q^1, ..., p^{n-1} - q^{n-1}, q^n - p^n\}$, then $p^n \succ q^n$.
\end{axiom}

\begin{proposition} \label{thm_disc_iff_1}
    $D$ can be rationalized by EU if and only if it satisfies Axiom \ref{axiom_disc_mod}.
\end{proposition}
\begin{proof}
    Necessity follows immediately by contradiction from Gordan's lemma. As for sufficiency, suppose a size-$n$ dataset $D \equiv \{p_1 \succ q_1, ..., p_n \succ q_n\}$ satisfies the axiom. If $\textbf{0} \not \in \conv\{p_1 - q_1, ..., p_{n-1}-q_{n-1}, p_n -q_n\}$, then Gordan's lemma delivers the representation. If $\textbf{0}  \in \conv\{p_1 - q_1, ..., p_{n-1}-q_{n-1}, p_n -q_n\}$, the axiom  requires $q_n \succ p_n$, a contradiction.
\end{proof}

The following result says that Axiom \ref{axiom_disc_mod} is equivalent to $n$-strong independence, which is the first assertion of Theorem \ref{thm:EU}. 

\begin{proposition} \label{thm_nSI}
    $D$ can be rationalized by EU if and only if it satisfies Axiom \ref{ax:n-SI}.
\end{proposition}
\begin{proof}
    By Proposition \ref{thm_disc_iff_1}, it suffices to show that when $p^1 \succ q^1, ..., p^{n-1} \succ q^{n-1}$,   $\textbf{0} \in \conv\{p^1 - q^1, ..., p^{n-1} - q^{n-1}, q^n - p^n\} \iff  p^n - q^n = \sum_{i=1}^{n-1} \lambda_i (p^i - q^i)$ for some $\lambda^i \geq 0$. Note that 
    \begin{align}
        \textbf{0} \not \in \conv\{p^1 - q^1, ..., p^{n-1} - q^{n-1}\}, \tag{$\ast$}
    \end{align}
    because if else, both $\textbf{0}  \in \conv\{p^1 - q^1, ..., p^{n-1} - q^{n-1}, q^n- p^n\}$ and $\textbf{0}  \in \conv\{p^1 - q^1, ..., p^{n-1} - q^{n-1}, p^n- q^n\}$ hold, and Axiom \ref{axiom_disc_mod} then requires $p^n \succ q^n$ and $q^n \succ p^n$.
    
    The desired claim then follows from noting that
    \begin{align}
         \sum_{i=1}^{n-1} \alpha_i (p^i - q^i) + \left( 1- \sum_{i=1}^{n-1}\alpha_i \right) (q^n - p^n)  = \textbf{0} &  \iff p^n - q^n = \sum_{i=1}^{n-1} \frac{\alpha_i}{ 1- \sum_{j=1}^{n-1}\alpha_j} (p^i - q^i), \nonumber  
    \end{align}
    where $1- \sum_{j=1}^{n-1}\alpha_j > 0$ by $(\ast)$.
\end{proof}

\subsection{Proof of the second assertion of Theorem \ref{thm:EU}}

We make use of the following mathematical fact, see, e.g. Theorem I.4.2 of \cite{Barvinok}.
\begin{lemma}[Helly]\label{lem:Helly}
Let $A_1, \ldots, A_n$ be a finite family of convex subsets of $\bR^d$. Suppose that $n>d$ and every $d+1$ of the sets have a common point:
\begin{equation}
    A_{i_1}\cap \cdots \cap A_{i_{d+1}}\neq \emptyset.\label{eq:Helly-1}
\end{equation}
Then all the sets have a common point:
\begin{equation}
A_{1}\cap \cdots \cap A_{n}\neq \emptyset.    \label{eq:Helly-2}
\end{equation}
\end{lemma}

Recall that $D$ satisfies condition \eqref{eq:intersection1} in Section \ref{sec:checking} iff $D$ is EU-consistent. Notice that condition \eqref{eq:intersection1} can be written as \eqref{eq:Helly-2} with $A_i=H^{(2y_i-1)(p_i-q_i)}$. Letting $d:=|Z|$, by Lemma \ref{lem:Helly} this is equivalent to \eqref{eq:Helly-1}
 for any $i_1, \ldots, i_{d+1}$. By Proposition \ref{thm_nSI} for any subset of $D$ of size $d+1$ condition \eqref{eq:Helly-1} is equivalent to $d+1$-strong independence. The last step is to notice that the sets $H^{p-q}$ can be treated as sets of normalized utilities $\{u\in \bR^Z: \sum_{z\in Z} u(z)=0\}$ and therefore  are isomorphic to subsets of $\bR^{d-1}$, in which case it suffices to check $d$-strong independence.\footnote{One could normalize utilities twice, taking $\{u\in \bR^Z: \sum_{z\in Z} u(z)=0, \sum_{z\in Z} u^2(z)=1\}$ as in  \cite{AhnSarver}, but this second normalization is nonlinear and does not preserve convexity of sets.}

\subsection{Proof of Corollary \ref{cor:random-erev}}

Let $Z$ be the set of all prizes of all the lotteries in the experiment. Because the distribution over prizes is continuous, with probability one all prizes will be distinct. Since all lotteries have distinct supports, all the vectors $d_i:=p_i-q_i$ have distinct supports. Thus, those $n$ vectors are linearly independent, so the premise of Axiom \ref{ax:n-SI} is not satisfied for any enumeration. In other words, the axiom is satisfied (vacuously) for every enumeration, regardless of the actual choices. By Theorem \ref{thm:EU}, the dataset is then EU-consistent with probability one.\qed

\subsection{Proof of Corollary \ref{cor:2}}

Let $D^*=\left((p_i, q_i, \varrho_i)\right)_{i=1}^n$  be the stochastic dataset. For any vector $y\in \{0, 1\}^n$ consider the dataset $D_y$ given by $\left((p_i, q_i, y_i)\right)_{i=1}^n$. By Corollary \ref{cor:random-erev}, $D_y$ is EU-consistent with probability one. Thus, the set of EU-consistent choices is $Y_n=\{0, 1\}^n$, and by Proposition \ref{prop:1}, $D^*$ is REU-consistent for any vector $\varrho$. \qed



\section{Proofs in the Stochastic Case}\label{app:REU}

A size-$n$ dataset is now $\left((p_i, q_i, \varrho_i)\right)_{i=1}^n$. Let $Y_n\subseteq \{0, 1\}^n$ be the set of EU-consistent deterministic choices between those lotteries. That  is, $y\in Y_n$ if there exists $u\in \bR^Z$ such that $\langle u, (2y_i-1)(p_i-q_i)\rangle>0$ for all $i=1, \ldots, n$.

\subsection{Proof of Proposition \ref{prop:1}}

\begin{proof}[Proof of Necessity]
    Suppose there exists some $\mu \in \Delta(\bR^Z)$ such that 
    \begin{align*}
    \varrho_i &= \mu\left(\{ u \in \bR^Z: \langle u, p_i \rangle > \langle u , q_i \rangle \} \right) \text{ for all }i \in \{1,...,n\}\\
    1&= \mu\left(\{ u \in \bR^Z: \langle u , p_i \rangle \neq  \langle u, q_i \rangle \text{ for all } i \in \{1,...,n\} \}\right).        
    \end{align*}
    Define the set
    $$A:=\supp(\mu)\cap \{u\in \bR^Z:  \langle u , p_i \rangle -  \langle u, q_i \rangle \neq 0 \text{ for all }i \in \{1,...,n\}\}.$$
    From  above,  we know that $\mu(A)=1$. Define the mapping $\Psi: A \to Y_n$ given by $\Psi(u) = y$ iff $\langle u,  (2y_i-1)(p_i - q_i)\rangle > 0$ for all $i \in \{1,...,n\}$. Define $\beta(y) := \mu\{u \in \bR^z : \Psi(u) = y\}$, and observe that
    \begin{align*}
        \varrho_i    & = \mu\{u \in \bR^Z : \langle u,  p_i - q_i \rangle >0\} \\
& = \mu \bigg \{u \in \bR^Z : \Psi(u) \in \{ y \in Y_n : y_i = 1\}  \bigg \}   \\
        & = \sum_{y \in Y_n: y_i = 1} \mu\{u \in \bR^Z : \Psi(u) = y\}  \\    
        & = \sum_{y \in Y_n: y_i = 1} \beta(y)  \\    
        &=\sum_{y \in Y_n} \beta(y) y_i,
    \end{align*}
    so $\rho\in \conv(Y_n)$ because $\beta\in\Delta(Y_n)$.
    
\noindent \emph{Proof of Sufficiency.} Suppose there exists some $\beta\in \Delta(Y_n)$ such that  $\varrho_i=\sum_{y \in Y_n} \beta(y) y_i$ for all $i=1, \ldots, n$. For each $y \in Y_n$, pick an arbitrary element $\bar{u}_y\in \Psi^{-1}(y)$ and define $\mu (\cdot) := \sum_{y \in Y_n} \beta(y) \delta_{\bar{u}_y}(\cdot)$, where  $\delta_{\bar{u}_y}(\cdot)$ is the point mass supported on $\{\bar{u}_y\}$. Observe that
\begin{align*}
    \mu \{ u \in \bR^Z: \langle u,  p_i - q_i \rangle >0\} = \sum_{\substack{y \in Y_n: \\  \langle \bar{u}_y,  p_i - q_i \rangle >0}} \beta(y)   = \sum_{y \in Y_n} \beta(y) y_i  = \varrho_i,
\end{align*}
where the second equality follows from the fact that $\langle \bar{u}_y, p_i - q_i \rangle >0 \iff y_i = 1$.
\end{proof}

\subsection{Proof of Necessity in Theorem \ref{thm:REU}}


\subsubsection{Preliminary Results}

In this subsection, we will develop some results for a single enumeration  in $k$-strong independence ($k$-SI). The next subsection will put them  together.  For the remainder of this subsection, we will fix an enumeration such that 
\begin{align*}		 
p^k-q^k&=\sum_{l=1}^{k-1} \lambda^l (p^l-q^l), \text{ for }\lambda_l\geq 0.
\end{align*}
In the dataset $D$, this corresponds to 
$p_i - q_i = \sum_{j \in J} \lambda_j (p_j - q_j)$ for some set $J$ with cardinality $k-1$ and $\lambda_j \in \bR \setminus \{0\}$. The weights $\lambda_j$ are negative when enumeration  flips $p$ and $q$. Define $J^+ \equiv \{j \in J : \lambda_j > 0\}$ and $J^- \equiv \{j \in J: \lambda_j < 0\}$.

For this particular enumeration, we will characterize the types of points $R$ that $k$-SI removes from $\{0,1\}^n$, and illustrate that removing $R$ from $\{0,1\}^n$ corresponds to imposing two linear constraints. 

\begin{definition}\label{def:R_m}
    Let $R$ be the deterministic choices that the application of Axiom \ref{ax:n-SI}  to the enumeration above removes from $\{0,1\}^n$. Let $G:=\{0,1\}^n\setminus  R$. 
\end{definition}
 
\begin{lemma} \label{lemma_prelim_SI1}
    $R  = R_1 \cup R_2$, where
    \begin{align*}
        R_1&:=\bigg\{ y \in \{0,1\}^n :  y_i = 0, \,\, y_j = \1_{ \{j \in J^+\} }  \,\, \forall j \in J \bigg\}\\
        R_2&:=\bigg\{ y \in \{0,1\}^n :  y_i = 1, \,\, y_j = \1_{ \{j \in J^-\} } \,\, \forall j \in J \bigg\}
    \end{align*}
\end{lemma}

\begin{proof}
       Observe that we can write
       \begin{align*}
           p_i - q_i = \sum_{j \in J} \lambda_j (p_j - q_j) = \sum_{j \in J^+} \lambda_j (p_j - q_j) + \sum_{j \in J^-} |\lambda_j| (q_j - p_j),
       \end{align*}
       such that all coefficients are positive.  $n$-SI then implies
       \begin{align*}
           p_j \succ q_j \text{ for all } j \in J^+ \,\,\, \text{and} \,\,\,  p_j \prec q_j \text{ for all } j \in J^- & \spimp p_i \succ q_i \\
            p_j \prec q_j \text{ for all } j \in J^+ \,\,\, \text{and} \,\,\,   p_j \succ q_j \text{ for all } j \in J^- &  \spimp p_i \prec q_i.
       \end{align*}
       The first scenario rules out any $y \in \{0,1\}^n$ satisfying $y_j = \1_{ j \in J^+ }$ for all $j \in J$ and $y_i = 0$. The second scenario rules out any $y \in \{0,1\}^n$ satisfying $y_j = \1_{ j \in J^-}$ for all $j \in J$ and $y_i = 1$. 
    \end{proof}

\begin{lemma} \label{lemma_prelim_SI2}
    \begin{align}
     G \subseteq \left\{ y \in \{0,1\}^n :  -|J^-| \leq \sum_{j \in J^+} y_j - \sum_{j \in J^-} y_j - y_i \leq |J^+| - 1   \right\}. \nonumber
    \end{align}
\end{lemma}
\begin{proof}
Suppose $y \in G$, and assume, for a contradiction, that $y \not \in$ RHS. This implies that either
\begin{align}
    \sum_{j \in J^+} y_j - \sum_{j \in J^-} y_j - y_i \leq   -|J^-|-1 \hspace{10pt} \text{or} \hspace{10pt} \sum_{j \in J^+} y_j - \sum_{j \in J^-} y_j - y_i  \geq |J^+| \nonumber
\end{align}

The former scenario can only happen if $y_j = 0$ for all $j \in J^+$, $y_j = 1$ for all $j \in J^-$, and $y_j = 1$. This is a contradiction, since all such points are in $R_2$.  The latter scenario can only happen if $y_j = 1$ for all $j \in J^+$, $y_j = 0$ for all $j \in J^-$ and $y_i = 0$. This is a contradiction, since all such points are in $R_1$. This proves $G \subseteq RHS$, as desired. 
\end{proof}
 
\begin{lemma} \label{lemma_convex_closure}
\begin{equation*}
        \conv(G) \subseteq \left\{ \phi \in [0,1]^n :- | J^-| \leq \sum_{j \in J^+} \phi_j - \sum_{j \in J^-} \phi_j - \phi_{i} \leq |J^+| - 1      \right\}
\end{equation*}
\end{lemma}
\begin{proof}   
    Every $y \in G$ satisfies the constraint  $-|J^-| \leq \sum_{j \in J^+} y_j - \sum_{j \in J^-} y_j - y_i \leq |J^+| - 1$ , and the convex combination of any such $y$ continues to satisfy the constraint.
\end{proof}

It can be proved that the set inclusion in Lemma \ref{lemma_convex_closure} is actually an equality. 
 We show this in Section \ref{app:proofsplus} of the Online Appendix.  However, as we show later, this does not help us prove sufficiency of $n$-triangulation.
 



\subsubsection{The actual proof of necessity}

If a given size $n$-dataset $D$ is REU-consistent, Proposition \ref{prop:1} implies that $\varrho \in \conv(Y_n)$. The premise of $k$-triangulation for every $k \in \{2,...,n\}$ is satisfied a finite number of times; call this number $M$. For every $m \in \{1,...,M\}$, let $i_m$  and $J_m$ denote the corresponding index and index set, such that $p_{i_m} - q_{i_m} = \sum_{j \in J_m} \lambda_j(p_j - q_j)$ for some $\lambda \in (\bR\setminus \{0\})^{ J_m} $. For each $m$, define the set $G_m$ as in Definition \ref{def:R_m}. Because $n$-SI implies $k$-SI for $k\leq n$, by Proposition \ref{thm_nSI} we have that
\begin{align*} 
    Y_n & = \bigcap_{m =1}^M G_m.
\end{align*}

Recall  that for sets $A,B$ in Euclidean space,   $\conv(A \cap B)\subseteq \conv(A) \cap \conv(B)$. Thus, if $D$ is REU-consistent, then  $\varrho \in \conv(Y_n)$, and thus 
$$\varrho\in \conv(G_m)$$
 for all $m=1, \ldots, M$.

By Lemma \ref{lemma_convex_closure}, this implies that for each $m$, 
   \begin{align*}
         - | J_m^-| \leq \sum_{j \in J^+_m} \varrho_j - \sum_{j \in J^-_m} \varrho_j - \varrho_{i_m} & \leq |J_m^+| - 1  \\
        \shortintertext{iff}
          - | J_m^-| \leq \sum_{j \in J^+_m} \varrho(p_j, q_j) - \sum_{j \in J^-_m} \varrho(p_j, q_j) - \varrho(p_{i_m}, q_{i_m} ) & \leq |J_m^+| - 1  \\
         \shortintertext{iff}
         0 \leq \sum_{j \in J^+_m} \varrho(p_j, q_j)  + \sum_{j \in J^-_m}    \varrho(q_j, p_j) - \varrho(p_{i_m}, q_{i_m} ) &\leq |J_m| - 1\\
                 \shortintertext{iff}
          \varrho(p_{i_m}, q_{i_m} ) \leq \sum_{j \in J^+_m} \varrho(p_j, q_j)  + \sum_{j \in J^-_m}    \varrho(q_j, p_j) & \leq  \varrho(p_{i_m}, q_{i_m} ) + |J_m| - 1,
                 \shortintertext{iff}
         \sum_{j \in J^+_m} \varrho(p_j, q_j)  + \sum_{j \in J^-_m}    \varrho(q_j, p_j)  - (|J_m| -1)   \leq  \varrho(p_{i_m}, q_{i_m}) & \leq \sum_{j \in J^+_m} \varrho(p_j, q_j)  + \sum_{j \in J^-_m}    \varrho(q_j, p_j),  
   \end{align*}
   where the second step follows by adding $|J_m^-|$ to each side of the inequality and substituting $\varrho(q_j, p_j)$ for $1-\varrho(p_j, q_j)$. The final relation is exactly what $|J_m|+1$-triangulation requires in the case where $p_{i_m}-q_{i_m} = \sum_{j\in J_m^+} \lambda_j (p_j -q_j) + \sum_{j\in J_m^-} |\lambda_j| (q_j -p_j)$, so we have the desired result.

\subsection{Proof of Sufficiency in Theorem \ref{thm:REU} ($n \leq 4$)}




\subsubsection{Equivalent datasets}

Let $D$ be a dataset of size $n$. For any $k$-enumeration of $D$, we can construct a new dataset $\hat D= \{ (p^1, q^1, \varrho^1), \ldots, (p^k, q^k, \varrho^k)  \}$. We say that the enumeration is \textit{exact} if  $k=n$ and every pair of lotteries is featured the same number of times in $D$ and $\hat D$ (possibly flipping the order). 

If the enumeration is exact, we say that sets $D$ and $\hat D$ are \textit{equivalent}. Any two datasets $D$ and $D'$ are equivalent if there is a exact enumeration such that $D'=\hat D$. Note that if two sets are equivalent, then they are both either REU-consistent or REU-inconsistent (this follows because ties have $\mu$-zero probability).

\subsubsection{Proof of Sufficiency for $n=2$}\label{sec:suff-2}

There are two cases to consider:

\paragraph{Case 1} The premise of 2-SI does not hold for any enumeration of the dataset. Then by Theorem \ref{thm:EU}, $Y_2 = \{0,1\}^2$ and $\varrho \in \conv\{Y_2\} = [0,1]^2$, as desired.

\paragraph{Case 2} The premise of 2-SI holds for some enumeration, such that $\hat p^2 - \hat q^2 = \lambda (\hat p^1 - \hat q^1)$ with $\lambda > 0$. Notice that then $\hat p^2 - \hat q^2$ cannot be a positive multiple of $\hat q^1 - \hat p^1$. Consider the equivalent dataset $\hat D:= \{ (\hat p^1, \hat q^1, \hat \varrho^1) , (\hat p^2, \hat q^2, \hat \varrho^2) \}$. The associated set $\hat Y_2$ equals $\{(0, 0),  (1,1)\}$, and  2-triangulation  requires $\hat \varrho^1 = \hat \varrho^2$. Since $\conv(\hat Y_2) = \{\hat \varrho\in [0,1]^2 : \hat \varrho^1 = \hat \varrho^2 \}$, it follows that $\hat \varrho \in \conv(\hat Y_2)$, so by Proposition \ref{prop:1} the set $\hat D$ is REU-consistent. Since $D$ and $\hat D$ are equivalent, the dataset $D$ is also REU-consistent.\\

For completeness, we should check what happens when the premise of 2-SI holds for some enumeration, such that $\hat p^2 - \hat q^2 = \lambda (\hat p^1 - \hat q^1)$ with $\lambda = 0$. But this implies that $\hat p^2=\hat q^2$, which is ruled out by the requirement that $p_i\neq q_i$ for all $i=1, \ldots, n$. Thus, the two above cases are exhaustive, completing the proof.\qed

\subsubsection{Facts about cones}\label{sec:cones}

To deal with bigger data sets, it is helpful to first develop some preliminary results. Given the dataset $\{(p_i, q_i, \varrho_i)\}_{i=1}^n$, define $d_i \equiv p_i - q_i$.  Let $\cones\{\cdot\}$ refer to the strict conic inclusion relation, such that $v \in \cones\{v', v''\}$ implies $v = \alpha v' + \beta v''$ with $\alpha, \beta > 0$.

\begin{lemma} \label{lemma_signs}
    For any $n \geq 2$ and $d_1,...,d_n \in \bR^m$, if  $d_n \in \cones\{ d_{i} \}_{i=1}^{n-1}$ and $\rank(d_1,...,d_n) = n-1$, then the only conic inclusions of the form $\pm d_i \in \cones\{\pm d_j : j \neq i \}$ are exactly the following $2n$ inclusions:
    \begin{align*}
        d_n \in \cones\{d_1,...,d_{n-1}\}, & \hspace{20pt} -d_n \in \cones\{-d_1,...,-d_{n-1}\} \\
        d_i \in \cones\{ (-d_j)_{j \neq n, i}, d_n \}, & \hspace{20pt } -d_i \in \cones\{(d_j)_{j \neq n, i}, -d_n \}.
    \end{align*}
\end{lemma}
\begin{proof}
    Since $d_n \in \cones\{d_i\}_{i=1}^{n-1}$, there exists $\alpha_i > 0$ such that
    \begin{align*}
        \sum_{i=1}^{n-1} \alpha_i d_i - d_n = 0.
    \end{align*}
    Since $\rank(d_1,...,d_n) = n-1$, we know that $( \alpha_1, ..., \alpha_{n-1}, -1)$ is the unique set of coefficients  (up to a scalar multiple) for which the above relation holds. For $\epsilon_i \in \{-1,1\}$, suppose that we have
    \begin{align*}
        \epsilon_i d_i \in \cones\{ \epsilon_j d_j\}_{j \neq i} \spiff (- \epsilon_i )d_i + \sum_{j \neq i} ( \epsilon_j \gamma_j ) d_j = 0
    \end{align*}
    for $\gamma_j > 0$. If $i = n$, the above result implies that
    \begin{align*}
        (   \epsilon_1 \gamma_1,  \epsilon_2 \gamma_2 , ...,  - \epsilon_n) = t (\alpha_1, ..., \alpha_{n-1}, -1)
    \end{align*}
    for some $t \neq 0$. Equivalently, the two  conic inclusions we have are  $d_n \in \cones\{d_1,...,d_{n-1}\}$  (if $t>0$) and $-d_n \in \cones \{-d_1, ..., -d_{n-1}\}$ (if $t < 0)$.

    If $i < n$ (take $i = n-1$ without loss), the above result implies that
    \begin{align*}
         (\epsilon_1 \gamma_1, ..., \epsilon_{n-2} \gamma_{n-2}, - \epsilon_{n-1}  ,\epsilon_n \gamma_n  ) = t (\alpha_1, ..., \alpha_{n-1}, -1)
    \end{align*}
    for some $t \neq 0$. Equivalently, the two conic inclusions we have are $d_{n-1} \in \cones\{(-d_j)_{j \leq n-2},  d_n \}$ (if $t <0)$ and  $-d_{n-1} \in \cones\{(d_j)_{j \leq n-2}, - d_n \}$ (if $t > 0$), as desired. 
\end{proof}

Here is a corollary that will be useful in the $n=3$ case.

\begin{corollary} \label{cor_3SI}
    If $\rank(d_1,d_2,d_3) = 2$ and the premise of 2-SI does not hold, then 
    \begin{align*}
        (1,1,0) \not \in Y_{3} \spimp Y_{3} = \{0,1\}^3 \setminus \{  (1,1,0) ,  (0,0,1) \}.
    \end{align*}
\end{corollary}
\begin{proof}
   Since the premise of 2-SI does not hold,   $(1,1,0)$ is removed from $\{0,1\}^3$ by 3-SI with strictly positive coefficients, i.e., $d_3 \in \cones\{d_1,d_2\}$. Lemma \ref{lemma_signs} and  3-SI then implies
\begin{align*} 
d_1 \in \cones\{-d_2, d_3\} \spimp (0,0,1) \not \in Y_3, \hspace{20pt} - d_1 \in \cones\{d_2, -d_3\} & \spimp (1,1,0) \not \in Y_3, \\ 
d_2 \in \cones\{-d_1, d_3\} \spimp (0,0,1) \not \in Y_3, \hspace{20pt} -d_2 \in \cones\{d_1, -d_3\} & \spimp (1,1,0) \not \in Y_3, \\ 
\hspace{20pt} - d_3 \in \cones\{-d_1, -d_2\} & \spimp (0,0,1) \not \in Y_3. 
\end{align*}    
    Since no other conic inclusion relations hold, no other instances of the premise of Strong Independence hold, in which case we can conclude that $Y_3 = \{0,1\}^3 \setminus \{  (1,1,0) ,  (0,0,1) \}$, as desired.
\end{proof}

Likewise, here is a corollary that will be useful in the $n=4$ case.

\begin{corollary}\label{cor_4SI}
    If $\rank(d_1,d_2,d_3,d_4) = 3$ and the premise of $2,3$-Strong Independence does not hold, then
    \begin{align*}
        (1,1,1,0) \not \in Y_4 \spimp Y_4 = \{0,1\}^4 \setminus \{ (1,1,1,0), (0,0,0,1) \}. 
    \end{align*}
\end{corollary} 

\begin{proof}
    Since the premise of 2,3-SI does not hold, $(1,1,1,0)$ is removed from $\{0,1\}^4$ by 4-SI with strictly positive coefficients, i.e., $d_4 \in \cones\{d_1, d_2,d_3\}$. Lemma \ref{lemma_signs} and 4-SI then implies
    \begin{align*} 
        d_1 \in \cones\{ -d_2, -d_3, d_4\} &\spimp (0,0,0,1) \not \in Y_4, \\ 
        -d_1 \in \cones\{ d_2, d_3, -d_4\} &\spimp (1,1,1,0) \not \in Y_4, 
    \end{align*}    
    and likewise by symmetry for $d_2,d_3$ as well. Since no other conic inclusion relations hold, no other instances of the premise of Strong Independence hold, in which case we can conclude that $Y_4 = \{0,1\}^4 \setminus \{  (1,1,1,0), (0,0,0,1) \}$, as desired.
\end{proof}

\subsubsection{Proof of Sufficiency for  $n=3$}\label{sec:suff-3}
There are three cases to consider:

\paragraph{Case 1} The premise of 3-SI does not hold for any enumeration of the dataset. Then, by Theorem \ref{thm:EU}, $Y_3 = \{0,1\}^3$ and $\varrho \in \conv\{Y_3\} = [0,1]^3$, satisfying the convex hull condition.

 \paragraph{Case 2} There is an enumeration $\{(p^1, q^1, \varrho^1), (p^2, q^2, \varrho^2)\}$ such that the premise of 2-SI holds, i.e., $d^2=\lambda d^1$ for some $\lambda>0$. To complete the enumeration, let $(p^3, q^3, \varrho^3)$ be the third element in the dataset. Consider the equivalent dataset $\hat D$ given by this exact enumeration. 2-SI requires $y \in \hat Y_3$ (implies $y^1 = y^2$) and 2-triangulation requires $\varrho^1 = \varrho^2$. Consider now the  dataset $\ddot D=\left((p^2, q^2, \varrho^2), (p^3, q^3, \varrho^3)  \right)$. Note that the corresponding set $\ddot Y_2$ equals  $\hat Y_3$ with the first dimension removed. It follows that $\varrho \in \conv(\hat Y_3)$ iff $(\varrho^2, \varrho^3) \in \conv(\ddot Y_2)$.  Thus, by Proposition \ref{prop:1}, $\hat D$ is REU-consistent iff  $\ddot D$ is REU-consistent. The result from Section \ref{sec:suff-2} guarantee that 2-triangulation is sufficient for $\ddot D$ to be REU-consistent.

 \paragraph{Case 3}  The premise of 2-SI does not hold for any enumeration and there is an enumeration such that the premise of 3-SI \begin{equation}
    d^3=\lambda^1 d^1+\lambda^2 d^2\label{eq:locnum1}
\end{equation} 
 holds with  $\lambda^1, \lambda^2>0$. This implies that $\rank(d^1, d^2, d^3) = 2$.  Consider the equivalent dataset $\hat D=\left((p^1, q^1, \varrho^1), \ldots, (p^3, q^3, \varrho^3)  \right)$. Equation \eqref{eq:locnum1} implies that $(1,1,0) \not \in \hat Y_3$.  Corollary \ref{cor_3SI} then implies  $\hat Y_3 = \{0,1\}^3 \setminus \{(1,1,0), (0,0,1) \}$,  and 3-triangulation requires $\varrho_1 + \varrho_2 - 1 \leq \varrho_3 \leq \varrho_1 + \varrho_2 \iff 0 \leq \varrho_1 + \varrho_2 - \varrho_3 \leq 1$. It can be verified directly (or seen as a consequence of our Lemma \ref{lemma_convex_closure-strong}) that
\begin{align*}
    \conv(\hat Y_3) = \{ r \in [0,1]^3 : 0 \leq r_1 + r_2 - r_3 \leq 1 \},
\end{align*}
so $\varrho \in \conv(\hat Y_3)$ and $\hat D$ is REU-consistent, so $D$ is REU-consistent. \qed

\subsubsection{Proof of Sufficiency for  $n = 4$}\label{sec:suff-4}
There are four cases to consider:

\paragraph{Case 1} The premise of 4-SI does not hold for any enumeration. Then, $Y_4 = \{0,1\}^4$ and $\varrho \in \conv\{Y_4\} = [0,1]^4$, satisfying the convex hull condition.

\paragraph{Case 2} There is an  enumeration $\{(p^1, q^1, \varrho^1), (p^2, q^2, \varrho^2)\}$ such that the premise of 2-SI holds.   To complete the enumeration, let $(p^3, q^3, \varrho^3)$ and $(p^4, q^4, \varrho^4)$ be the third and fourth items in the dataset. Consider the equivalent dataset $\hat D$ given by this enumeration. We have $y \in \hat Y_4$ implies $y_1 = y_2$,  and 2-triangulation requires $\varrho^1 = \varrho^2$. 

Consider the dataset $\ddot D:=\{(p^2, q^2, \varrho^2), (p^3, q^3, \varrho^3), (p^4, q^4, \varrho^4)\}$. The associated set $\ddot Y_3$  equals to  $\hat Y_4$ with the first dimension removed. It follows that $\varrho \in \text{conv}(\hat Y_4)$ iff $(\varrho^2, \varrho^3, \varrho^4) \in \conv(\ddot Y_3)$. Thus, by Proposition \ref{prop:1}, $\hat D$ is REU-consistent iff $\ddot D$ is REU-consistent. The result from section \ref{sec:suff-3} guarantees that 3-triangulation is sufficient for $\ddot D$ to be REU-consistent.

\paragraph{Case 3} The premise of 2-SI does not hold for any enumeration and there is an enumeration such that the premise of 3-SI \eqref{eq:locnum1} 
 holds with  $\lambda^1, \lambda^2>0$. Thus, $d^3\in \cones \{d^1, d^2\}$ and $\rank(d^1, d^2, d^3)=2$. To complete the enumeration, let $(p^4, q^4, \varrho^4)$ be the fourth item in the dataset and consider the equivalent dataset $\hat D$. We have $\rank(d^1, d^2, d^3, d^4) \in \{2,3\}$. Thus, we have two sub-cases:

\paragraph{Case 3a} Suppose $\rank(d^1, d^2, d^3, d^4) = 3$. Since no two $d^i$ are collinear, it follows that $d^4 \not \in \spann\{d^1,d^2\} =:E$, while $d^1,d^2,d^3 \in E$. Thus,   $4 \in \{i,j,k\}$ implies
\begin{align*}
    \epsilon_i d^i  \not \in \cones\{ \epsilon_j d^j, \epsilon_k d^k\}
\end{align*}
for any $\epsilon_i, \epsilon_j, \epsilon_k \in \{-1,1\}$.\footnote{If $i = 4$, this would imply $d^4 \in E$, a contradiction. If $k = 4$, this would imply $d^i \in \spann\{d^j, d^4\}$ and thus $d^4 \in E$ because $d^i, d^j \in E$, a contradiction. The same argument holds for $j = 4$. } That is, the premise of 3-SI never holds with the fourth lottery pair  $\{p^4, q^4\}$ as one of the arguments.

By the same logic, the premise of 4-SI never holds with the fourth lottery pair as one of the arguments because $4 \in \{c,i,j,k\}$ implies
\begin{align*}
    \epsilon_c d^c \not \in \cones\{ \epsilon_i d^i,  \epsilon_j d^j, \epsilon_k d^k\}
\end{align*}
for any $\epsilon_c, \epsilon_i, \epsilon_j, \epsilon_k \in \{-1,1\}$.\footnote{If $c= 4$, this would imply $d_4 \in E$, a contradiction. If $i = 4$, this would imply $d^c \in \spann\{d^4, d^j, d^k\}$ and thus $d_4 \in E$ because $d^j, d^k, d^c \in E$, a contradiction. The same argument holds for $j,k = 4$.}  This allows us to conclude that the fourth lottery pair plays no role in the construction of $\hat Y_4$.

It then follows from Case 3 in Section \ref{sec:suff-3} that $\hat Y_4 = \{0,1\}^4 \setminus \{ (1,1,0, \cdot) , (0,0,1, \cdot) \}$, such that
\begin{align*}
    \conv(\hat Y_4) = \{ r \in [0,1]^4 : 0 \leq r_1 + r_2 - r_3 \leq 1 \}.
\end{align*}
Since 3-triangulation requires $\varrho_1 + \varrho_2 - 1 \leq \varrho_3 \leq \varrho_1 + \varrho_2 \iff 0 \leq \varrho_1 + \varrho_2 - \varrho_3 \leq 1$, we have that $\varrho \in \conv(\hat Y_4)$, as desired.

\paragraph{Case 3b} Suppose $\rank(d^1,d^2,d^3,d^4) = 2$. As in Figure \ref{fig:vec-4}, we can  relabel menus and lotteries without loss, such that $d^2, d^3 \in \cones \{d^1, d^4\}$.\footnote{Since $\rank(d^1,d^2,d^3,d^4) = 2$, all vectors lie on a 2D plane. Take any $u \in \spann(d^1,d^2,d^3,d^4)$ satisfying $\langle u, d^i \rangle \neq 0$, and flip signs (swap $p$s and $q$s) for any $d^i$ such that $\langle u , d^i \rangle> 0$. This leads to all $d^i$ lying on the same halfspace. Order these by the value of the normalized inner product with vector $(1, 0)$), and re-define $d^1$ and $d^4$ as the vectors with the smallest and largest angles and $d^2,d^3$ as the vectors with the second and third smallest angles, such that $d^2,d^3 \in \cones\{d^1,d^4\}$.}

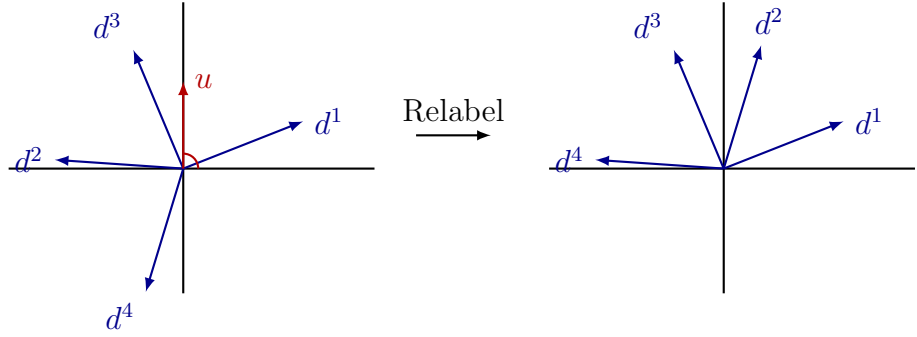
\begin{figure}[h!]
\begin{center}
\begin{tikzpicture}[
    >=latex,
    scale=1.1,
    vec/.style={->, thick, blue!55!black},
    axis/.style={thick},
    redvec/.style={->, thick, red!70!black}
]

\def\r{1.55}

\begin{scope}[shift={(0,0)}]
    \draw[axis] (-2.1,0) -- (2.3,0);
    \draw[axis] (0,-1.5) -- (0,2.0);

    \draw[vec] (0,0) -- (21.5:\r) node[right] {$d^1$};
    \draw[vec] (0,0) -- (112.8:\r) node[above left] {$d^3$};
    \draw[vec] (0,0) -- (176.1:\r) node[left] {$d^2$};
    \draw[vec] (0,0) -- (253.1:\r) node[below left] {$d^4$};

    \draw[redvec] (0,0) -- (0,1.05) node[right] {$u$};

    \draw[red!70!black, thick] (0.18,0) arc[start angle=0,end angle=90,radius=0.18];
\end{scope}

\node at (3.25,0.70) {Relabel};
\draw[->, thick] (2.80,0.40) -- (3.70,0.40);

\begin{scope}[shift={(6.5,0)}]
    \draw[axis] (-2.1,0) -- (2.3,0);
    \draw[axis] (0,-1.5) -- (0,2.0);

    \draw[vec] (0,0) -- (21.5:\r) node[right] {$d^1$};
    \draw[vec] (0,0) -- (112.8:\r) node[above left] {$d^3$};
    \draw[vec] (0,0) -- (176.1:\r) node[left] {$d^4$};
    \draw[vec] (0,0) -- (73.1:\r)
        node[pos=1, above, xshift=3pt, yshift=0pt] {$d^2$};
\end{scope}

\end{tikzpicture}

\caption{Lottery differences in case 3b of the $n=4$ sufficiency proof.} \label{fig:vec-4}
	\end{center}
\end{figure}

 3-SI then requires $p^1 \succ q^1, p^4 \succ q^4 \implies p^2 \succ q^2$ and  also $p^1 \succ q^1, p^4 \succ q^4 \implies p^3 \succ q^3$.  With datasets $\ddot D = \{ (p^1, q^1, \varrho^1), (p^2, q^2, \varrho^2), (p^4, q^4, \varrho^4)\}$,  $\ddot D' = \{ (p^1, q^1, \varrho^1), (p^3, q^3, \varrho^3), (p^4, q^4, \varrho^4)\}$ and corresponding $\ddot Y_3, \ddot Y_3'$, Corollary \ref{cor_3SI} (with the $d_i$'s appropriately relabeled)  then implies
\begin{align*}
    \ddot Y_3 = \ddot Y_3' = \{0,1\}^3 \setminus \{(1,0,1), (0,1,0)\}.
\end{align*}
We also see from Figure \ref{fig:vec-4} that $d^3 \in \cones\{d^2, d^4\}$ and $d^2 \in \cones\{d^1, d^3\}$, in which case 3-SI requires $p^2 \succ q^2, p^4 \succ q^4 \implies p^3 \succ q^3$ and also $p^1 \succ q^1, p^3 \succ q^3 \implies p^2 \succ q^2$. With datasets $\ddot D'' = \{ (p^2, q^2, \varrho^2), (p^3, q^3, \varrho^3), (p^4, q^4, \varrho^4) \}$, $\ddot D''' = \{ (p^1, q^1, \varrho^1), (p^2, q^2, \varrho^2), (p^3, q^3, \varrho^3) \}$ and corresponding $\ddot Y_4'', \ddot Y_4'''$, Corollary \ref{cor_3SI} (with the $d_i$'s appropriately relabeled) then implies
\begin{align*}
    \ddot Y_3'' = \ddot Y_3''' = \{0,1\}^3 \setminus \{(1,0,1), (0,1,0)\}.
\end{align*}

Importantly, observe that because the $d^i$ lie on a 2D plane, all instances of 4-SI are nested by 3-SI. This means that $\ddot Y_3$ corresponds to $Y_4$ with the third dimension removed, and likewise for $\ddot Y_3'$, $\ddot Y_3''$,  $\ddot Y_4''$  with the second, first, and fourth dimensions. In other words, we have that 
\begin{align*}
    Y_4 = \{0,1\}^4 \setminus \bigg \{ (1,0, \cdot, 1), (0,1, \cdot, 0),  (1, \cdot, 0, 1), (0, \cdot, 1, 0),  (\cdot, 1,0, 1), (\cdot,  0,1, 0),  (1,0, 1, \cdot), (0,1, 0, \cdot) \bigg \}.
\end{align*}

It can be verified that
\begin{align*} 
\conv(Y_4) & = \left\{ r \in [0,1]^4: 0 \leq \begin{array}{c} r_1 + r_4 - r_2 \\ r_1 + r_4 - r_3 \\ r_2 + r_4 - r_3 \\ r_1 + r_3 - r_2 \end{array} \leq 1 \right\}, \end{align*}
by directly enumerating the vertices of the polytope given by the linear inequalities.\footnote{Since this polytope is  in $\mathbb{R}^4$, each vertex is determined by a collection of four linearly independent active inequalities. Enumerating these collections, solving the resulting systems of equalities, and retaining only the feasible solutions yields precisely the points in $Y_4$. } 

On the other hand, 3-triangulation requires 
\begin{align*}
    \varrho^1 + \varrho^4 -1 \leq \varrho^2 \leq \varrho^1 + \varrho^4 & \spiff 0 \leq \varrho^1 + \varrho^4 - \varrho^2 \leq 1 \\
    \varrho^1 + \varrho^4 -1 \leq \varrho^3 \leq \varrho^1 + \varrho^4 & \spiff 0 \leq \varrho^1 + \varrho^4 - \varrho^3 \leq 1 \\ 
    \varrho^2 + \varrho^4 -1 \leq \varrho^3 \leq \varrho^2 + \varrho^4 & \spiff 0 \leq \varrho^2 + \varrho^4 - \varrho^3 \leq 1 \\
    \varrho^1 + \varrho^3 -1 \leq \varrho^2 \leq \varrho^1 + \varrho^3 & \spiff 0 \leq \varrho^1 + \varrho^3 - \varrho^2\leq 1.
\end{align*}
We can thus conclude that $\varrho \in \conv(Y_4)$, as desired.\\[-1em]

\paragraph{Case 4} The remaining scenario is when the premise of 4-SI holds only with strictly positive coefficients, in which case $\rank(d_1,d_2,d_3,d_4) = 3$. WLOG, relabel menus and lotteries such that 4-SI requires $p_1 \succ q_1$, $p_2 \succ q_2$, $p_3 \succ q_3 \implies p_4 \succ q_4$, which means that $(1,1,1,0) \not \in Y_4$. Corollary \ref{cor_4SI} then implies $Y_4 = \{0,1\}^4 \setminus \{(1,1,1,0) , (0,0,0,1) \}$, and 4-triangulation requires $\varrho_1 + \varrho_2 + \varrho_3 - 2 \leq \varrho_4 \leq \varrho_1 + \varrho_2 + \varrho_3 \iff 0 \leq \varrho_1 + \varrho_2 + \varrho_3 - \varrho_4 \leq 2$. It can be verified directly (or seen as a consequence of Lemma \ref{lemma_convex_closure-strong}) that  
\begin{align*}
    \conv(Y_4) = \{ r \in [0,1]^4: 0 \leq r_1 + r_2 + r_3 - r_4 \leq 2 \},
\end{align*}
so $\varrho \in \conv(Y_4)$, as desired.

\subsection{Counterexample: Axiom \ref{ax:n-triang}  not sufficient when $n\geq 5$}

Let $|X| = 3$, and consider the lotteries given by
\begin{align*}
    p_1 = \frac{1}{12} \left( 4, 5, 3 \right) ,  \hspace{20pt}p_2 = \frac{1}{12} \left( 3, 6 ,3 \right), \hspace{20pt}  p_3 = \frac{1}{12} \left(6,2,4 \right),  \\
    p_4 = \frac{1}{12} \left( 5, 6, 1  \right), \hspace{20pt}   p_5 = \frac{1}{12} \left( 6,5,1 \right),   \hspace{20pt} q = \frac{1}{12} \left( 4, 4, 4 \right),
\end{align*}
such that the 5 binary menus are  $\{p_i, q\}_{i=1}^5$.   Define the lottery differences $d_i \equiv p_i - q$, with
\begin{align*}
     d_1 = \frac{1}{12} (0, 1, -1),  \hspace{20pt} d_2 = \frac{1}{12} (-1, 2, -1),  &  \hspace{20pt} d_3 = \frac{1}{12} ( 2, -2, 0), \\
     d_4 = \frac{1}{12} (1,2,-3), &  \hspace{20pt}  d_5  = \frac{1}{12}  ( 2, 1, -3).
\end{align*}

Note that no two $d_i$ are collinear, implying that the premise of $2$-triangulation never holds, i.e.,  $2$-triangulation holds vacuously. In particular,  this corresponds to Case 3b in the $n=4$ proof, where there are many more combinations of vectors for which the premise of $k$-SI holds, for $k=3,4,5$ because now there are 5 vectors on a 2D plane. In the notation of the proof of necessity,  each of these correspond to the set $\conv(G_m) \subseteq [0,1]^5$ of vectors that satisfy the inequality conditions given by $k$-triangulation for $k=3,4,5$. Recall that $\bigcap_m \conv(G_m)$ is then the collection of stochastic choice rules that satisfy $k$-triangulation for $k = 2,3,4,5$ in this example.

The proof of necessity showed $\conv(Y_n) \subseteq \bigcap_m \conv(G_m)$ for general $n$. Here, the reason  sufficiency fails is because $\conv(Y_5) \subsetneq \bigcap_m \conv(G_m)$. Since $\conv(Y_5)$ has integer vertices by construction, one way to prove the validity of our counterexample is to show that $\bigcap_m \conv(G_m)$ has fractional vertices. We do so below by showing that the fractional vector 
\begin{align*}
   \varrho \equiv  (\varrho_1, \varrho_2, \varrho_3, \varrho_4, \varrho_5) = \frac{1}{3} (1,2,2,2,1),
\end{align*}
where $\varrho_i$ is the probability $p_i$ is chosen over $q$,  corresponds to a vertex of $\bigcap_m \conv(G_m)$.

First, the following lemma proves that our dataset satisfies $k$-triangulation for $k= 3,4,5$. Combined with the previous observation that $2$-triangulation holds vacuously, this then implies  $\varrho \in \bigcap_m \conv(G_m)$.

\begin{lemma} \label{lemma_counterex_tri}
    Fix any $r_i \in \{\varrho_i, 1-\varrho_i\}$ for each $i \in \{1,2,3,4,5\}$. For any $j \in \{1,2,3,4,5\}$ and $I \subseteq \{1,2,3,4,5\} \setminus \{j\}$ with $|I| \geq 2$,  we have
    \begin{align*}
        \sum_{i \in I} r_i - (|I|-1) \leq r_j \leq \sum_{i \in I} r_i.
    \end{align*}
\end{lemma}
\begin{proof}
   It suffices to show that 
   \begin{align*}
       \sum_{i \in I} 3r_i - 3(|I|-1) \leq 3r_j \leq \sum_{i \in I} 3r_i. \tag{$\ast$}
   \end{align*}
   Observe that $3r_i \in \{1,2\}$ for all $i$, in which case 
   \begin{align*}
       \sum_{i \in I} 3r_i - 3(|I|-1) & \leq 2|I| - 3 |I| + 3 = 3 - |I|  \underbrace{\leq}_{|I| \geq 2}  1  \leq 3r_j,
   \end{align*}
   so the lower end of $(\ast)$ always holds. It is also easy to see that the upper end of $(\ast)$ always holds, since $3 r_c + 3 r_k \geq 2$ for any $c,k \in \{1,2,3,4,5\}$, and $3r_j \in \{1,2\}$.
\end{proof}

Next, observe that the lottery differences satisfy the following conic relations:
\begin{align*}
    d_1 \in \cone\{d_2, d_3\}, \hspace{20pt} d_1 & \in \cone\{d_2, d_4\}  \tag{$d_1 = d_2 + \frac{1}{2}d_3 = \frac{1}{4}d_2 + \frac{1}{4}d_4$} \\
    d_5 \in \cone\{d_2,d_3\}, \hspace{20pt} d_5 & \in \cone\{d_3, d_4\} \tag{$d_5 =  3d_2 + \frac{5}{2}d_3 = \frac{1}{2}d_3 + d_4$} \\
    d_4 &  \in \cone\{d_1, d_5\}. \tag{$d_4 = \frac{3}{2}d_1 + \frac{1}{2}d_5$}
\end{align*}
Any $\varrho \in \bigcap_m \conv(G_m)$ must then satisfy the following inequalities obtained from 3-triangulation:
\begin{align*}
  \varrho_2 + \varrho_3 - 1 \leq  \varrho_1 \leq \varrho_2 + \varrho_3, \hspace{20pt} \varrho_2 + \varrho_4 - 1 &  \leq  \varrho_1 \leq \varrho_2 + \varrho_4 \\
  \varrho_2 + \varrho_3 - 1 \leq \varrho_5 \leq \varrho_2 + \varrho_3, \hspace{20pt}   \varrho_3 + \varrho_4 - 1 & \leq \varrho_5 \leq \varrho_3 + \varrho_4 \\
  \varrho_1 + \varrho_5  -1 & \leq \varrho_4 \leq \varrho_1 + \varrho_5.
\end{align*}
Consider the system of equations that arise when the  first four inequalities bind at the lower end and the last inequality binds at the upper end, such that
$$ \begin{pmatrix}
        1 \\ 1 \\ 1 \\ 1 \\ 0
    \end{pmatrix}  = \underbrace{\begin{pmatrix}
        -1 & 1 & 1 & 0 & 0 \\
        -1 & 1 & 0 & 1 & 0 \\ 
        0 & 1 & 1 & 0 & -1 \\
        0 & 0 & 1 & 1 & -1 \\
        1  & 0 & 0 & -1 &  1
    \end{pmatrix}}_{\equiv C} \varrho.$$
This system is solved uniquely by our $\varrho$. Since $\rank(C) = 5$ and $\bigcap_m \conv(G_m) \subseteq \bR^5$, the rank characterization of the vertices of polytopes \citep[e.g., Theorem II.4.2 of][]{Barvinok} imply that $\varrho$ is indeed a vertex of $\bigcap_m \conv(G_m)$.

Since $\varrho$ is a vertex of $\bigcap_m \conv(G_m)$ and it is not a vertex of $\conv(Y_5)$, it must be that these two sets are different. In particular, $\bigcap_m \conv(G_m)$ contains $\rho$ and $\conv(Y_5)$ does not, i.e., $\rho$ satisfies $k$ triangularity for $k=2, 3, 4, 5$ but is not REU-consistent.

\begin{spacing}{1.25}
\bibliographystyle{apalike}
\bibliography{EU.bib}
\end{spacing}

\newpage
\setcounter{section}{18}



\section{Supplementary Appendix}\label{app:supp}



\subsection{Further Exploration of the \texttt{choices13k} Dataset}\label{app:13k}

\subsubsection{The Dataset}

In the \texttt{choices13k} dataset, each question is a choice between \emph{Gamble A} and \emph{Gamble B}. Gamble $A$ is a lottery with binary support. Gamble $B$, however, is potentially compound: in the first stage, it either yields a fixed payoff or leads to a second-stage lottery. In some questions, which \cite{13k} classifies as ``ambiguous,'' subjects are not shown the probabilities associated with these first-stage events. In some questions, there exists non-zero correlation between the payoffs of the two options. The dataset also varies whether subjects receive feedback after making a decision. Each subject answers each question five times.

\subsection{Methodology behind Table \ref{tab:zeroprob}}

We restrict attention to the subset of choice scenarios that correspond to our simpler domain of choosing between (simple) lotteries. Operationally, this means dropping all questions involving ambiguity, correlation, or feedback, as well as all questions in which the second-stage lottery in Gamble $B$ is non-degenerate. This leaves $4{,}741$ subjects who answered at least two questions, which corresponds to the sample in the ``stochastic preference'' column of Table \ref{tab:zeroprob}. We construct the deterministic preference sample by dropping all subject--question observations for which the subject did not give the same answer in all five repetitions of the question. This leaves $2{,}183$ subjects who answered at least two questions, corresponding to the sample in the ``deterministic preference'' column.

\subsubsection{Violations as a function of the number of questions answered}

In Table \ref{tab:zeroprob}, each subject is one observation, although subjects differ in the number of distinct questions they answered.  Tables \ref{tab:zeroprob-1}-\ref{tab:zeroprob-2} 
divide the sample of subjects by the number of distinct questions answered ($n=2, 3, 4, 5, 6$) and report the proportion of actual and maximal violations of EU, monotone EU, and concave and monotone EU. 
 
\begin{table}[htbp]
\centering

\begin{tabular}{r|r|rr|rr|rr}
\toprule
 &  & 
\multicolumn{2}{c|}{\makebox[3.5cm][c]{EU}} & 
\multicolumn{2}{c|}{\makebox[3.5cm][c]{Mon. EU}} & 
\multicolumn{2}{c}{\makebox[3.5cm][c]{Mon. \& Conc. EU}} \\ 
$n$ & Subjects & Actual & Max & Actual & Max & Actual & Max \\
\midrule
2 & 1679 & 0.00 & 0.00 & 0.04 & 0.43 & 0.27 & 0.90 \\
3 & 426  & 0.00 & 0.00 & 0.06 & 0.59 & 0.33 & 0.98 \\
4 & 70   & 0.00 & 0.00 & 0.03 & 0.67 & 0.34 & 1.00 \\
5 & 6    & 0.00 & 0.00 & 0.00 & 0.67 & 0.67 & 1.00 \\
6 & 2    & 0.00 & 0.00 & 0.00 & 1.00 & 0.50 & 1.00 \\
\bottomrule
\end{tabular}
\caption{Proportion of representation violations by number of questions answered in the deterministic preference sample.}
\label{tab:zeroprob-1}
\end{table}

\begin{table}[htbp]
\centering

\begin{tabular}{r|r|rr|rr|rr}
\toprule
 &  & 
\multicolumn{2}{c|}{\makebox[3.5cm][c]{EU}} & 
\multicolumn{2}{c|}{\makebox[3.5cm][c]{Mon. EU}} & 
\multicolumn{2}{c}{\makebox[3.5cm][c]{Mon. \& Conc. EU}} \\ 
$n$ & Subjects & Actual & Max & Actual & Max & Actual & Max \\
\midrule
2 & 3185 & 0.00 & 0.00 & 0.07 & 0.36 & 0.33 & 0.88 \\
3 & 1287 & 0.00 & 0.00 & 0.12 & 0.56 & 0.47 & 0.99 \\
4 & 240  & 0.00 & 0.00 & 0.15 & 0.61 & 0.52 & 1.00 \\
5 & 25   & 0.00 & 0.00 & 0.24 & 0.76 & 0.56 & 1.00 \\
6 & 4    & 0.00 & 0.00 & 0.25 & 1.00 & 0.50 & 1.00 \\
\bottomrule
\end{tabular}
\caption{Proportion of representation violations by number of questions answered in the stochastic preference sample.}
\label{tab:zeroprob-2}
\end{table}

\subsubsection{Results for size-2 datasets}

 In Table \ref{tab:zeroprob-robust}, we consider an alternative exercise that focuses on size-2 datasets. After applying the filtering exercise described above, $954$ questions remain, corresponding to $\binom{954}{2}=454{,}581$ possible size-2 datasets. Because subjects were not assigned to every question, many of these size-2 datasets do not contain any pair of questions answered by the same subject. Requiring that at least one subject answered both questions leaves $2{,}146$ size-2 datasets. This produces a final sample size of $N=8{,}796$ subject $\times$ size-2 dataset observations for the stochastic preference column in Table \ref{tab:zeroprob-robust}. The deterministic preference counterpart is constructed the same way as before. 

\begin{table}[!h]
    \centering
    \begin{tabular}{l|cc|cc}
\toprule
  & \multicolumn{2}{c}{deterministic preference} & \multicolumn{2}{c}{stochastic preference} \\
\hline
Violations of ...      & actual & max   & actual & max   \\
\hline
EU                     & 0      & 0     & 0      & 0     \\
Monotone EU            & 121    & 1,492 & 673    & 3,358 \\
Monotone \& concave EU & 838    & 3,102 & 2,902  & 7,817 \\
\hline
$N$ & \multicolumn{2}{c}{3,467} & \multicolumn{2}{c}{8,796} \\
\bottomrule
\end{tabular}
\caption{We take the binary lotteries in the  \texttt{choices13k} dataset and consider every possible size-2 dataset. The ``deterministic preference'' column  restricts attention to questions that subjects have had consistent answers to across all 5 trials. The ``stochastic preference'' column defines a subject's deterministic choice as the option chosen more often.  The ``actual" column reports the actual number of violations in the data, and the  ``max" column reports the maximum possible number of violations. }
\label{tab:zeroprob-robust}
\end{table}

\subsection{Details of our Grid Search and Alternative Grids}\label{app:MRanomaliestwo}

We perform a crude grid approximation of $\Delta(Z)$, taking
$Z=\{0.01,2.505,5.0,7.495,9.99\}$ and using the randomly generated probability grid
$\{0.339,0.359,0.467,0.793,0.939\}$. 
For each pair of distinct prizes in $Z$, we assign probabilities $(p,1-p)$ with $p$ from this grid, yielding $5\binom{5}{2}=50$ binary-support lotteries. We then search over all unordered binary menus from these lotteries, for a total of $\binom{50}{2}=1225$ menus. 

In addition to MR's calibration, we consider four additional parameter values $\theta_4=(0.1, 1.55)$, $\theta_5=(0.1,2.275)$, $\theta_6=(0.1, 3)$, $\theta_7=(3,0.1)$ that are obtained from maximizing the number of anomalies and/or counterexamples over the generated menus.\footnote{For computational simplicity, we optimize over the coarse parameter space $\theta=(\delta,\gamma) \in \{0.1, 0.825, 1.55, 2.275, 3\}^2$.}

Our findings are robust to alternative approaches to menu construction. We consider two alternatives. The first  (Table \ref{tab:extratab_combined}) has an even probability grid $\{ 0.01, 0.255, 0.5, 0.745, 0.99 \}$. The second  (Table \ref{tab:extratab3-4-combined}) implements the algorithm of \cite{Erevetal17} discussed in Section \ref{sec:13k}. In particular, the finding that that our coarser categories capture almost all anomalies is robust to the gridding procedure, since we see that the  counts for Axiom \ref{ax:strong-LMR}  and Axiom \ref{ax:m-2-cancellation} are almost identical across specifications.

Finally, we note that $\theta_4$ and $\theta_7$ maximize the number of counterexamples and anomalies in the even gridding procedure, respectively, while $\theta_5$ and $\theta_6$ maximizes both the number of counterexamples and anomalies in the random gridding and \cite{Erevetal17} procedures, respectively.

\begin{table}[h!]
\centering
\begin{tabular}{l|ccc|cccc}
\toprule
 & \multicolumn{3}{c|}{MR parameters} & \multicolumn{4}{c}{optimized parameters} \\
 & $\theta_1$ & $\theta_2$ & $\theta_3$ & $\theta_4$ & $\theta_5$ & $\theta_6$ & $\theta_7$ \\
\hline
Violations of Axiom \ref{ax:dom}              & 0   & 0   & 0   & 102 & 96  & 94  & 0    \\
Violations of Axiom \ref{ax:SI}               & 0   & 0   & 0   & 0   & 0   & 0   & 0    \\
\hline
Dominated Consequence                         & 0   & 0   & 0   & 62  & 35  & 25  & 0    \\
Reverse Dominated Consequence                 & 0   & 0   & 0   & 54  & 34  & 24  & 0    \\
Strict Dominance                              & 0   & 0   & 0   & 51  & 28  & 18  & 0    \\
\hline
Violations of Axiom \ref{ax:Left}             & 29  & 16  & 16  & 358 & 235 & 191 & 264  \\
Violations of Axiom \ref{ax:Right}            & 37  & 16  & 14  & 504 & 384 & 335 & 238  \\
Violations of Axiom \ref{ax:Middle}           & 0   & 0   & 0   & 283 & 130 & 82  & 0    \\
\hline
Violations of Axiom \ref{ax:LMR}              & 56  & 32  & 28  & 827 & 621 & 558 & 571  \\
Violations of Axiom \ref{ax:strong-LMR}       & 253 & 106 & 101 & 1041 & 798 & 728 & 3091 \\
Violations of Axiom \ref{ax:m-2-cancellation} & 258 & 106 & 101 & 1045 & 802 & 732 & 3173 \\
\hline
\# Anomalies                                  & 258 & 106 & 101 & 1045 & 802 & 732 & 3173 \\
\bottomrule
\end{tabular}
\caption{Anomalies with MR ground truths (parameters $\theta_1$--$\theta_3$) and optimized ground truths ($\theta_4$--$\theta_7$) under even gridding, for 749{,}700 pairs of lotteries.}
\label{tab:extratab_combined}
\end{table}

\begin{table}[h!]
\centering
\begin{tabular}{l|ccc|cccc}
\toprule
 & \multicolumn{3}{c|}{MR parameters} & \multicolumn{4}{c}{optimized parameters} \\
 & $\theta_1$ & $\theta_2$ & $\theta_3$ & $\theta_4$ & $\theta_5$ & $\theta_6$ & $\theta_7$ \\
\hline
Violations of Axiom \ref{ax:dom}              & 0 & 0 & 0 & 21 & 20 & 21 & 0 \\
Violations of Axiom \ref{ax:SI}               & 0 & 0 & 0 & 0  & 0  & 0  & 0 \\
\hline
Dominated Consequence                         & 0 & 0 & 0 & 0  & 0  & 0  & 0 \\
Reverse Dominated Consequence                 & 0 & 0 & 0 & 0  & 0  & 0  & 0 \\
Strict Dominance                              & 0 & 0 & 0 & 0  & 0  & 0  & 0 \\
\hline
Violations of Axiom \ref{ax:Left}             & 0 & 0 & 0 & 0  & 0  & 0  & 0 \\
Violations of Axiom \ref{ax:Right}            & 0 & 0 & 0 & 0  & 0  & 0  & 0 \\
Violations of Axiom \ref{ax:Middle}           & 0 & 0 & 0 & 5  & 4  & 1  & 0 \\
\hline
Violations of Axiom \ref{ax:LMR}              & 0 & 0 & 0 & 62 & 152 & 156 & 25 \\
Violations of Axiom \ref{ax:strong-LMR}       & 0 & 0 & 0 & 62 & 152 & 156 & 25 \\
Violations of Axiom \ref{ax:m-2-cancellation} & 3 & 3 & 2 & 104 & 228 & 231 & 94 \\
\hline
\# Anomalies                                  & 3 & 3 & 2 & 104 & 228 & 231 & 94 \\
\bottomrule
\end{tabular}
\caption{Anomalies with MR ground truths (parameters $\theta_1$--$\theta_3$) and optimized ground truths (parameters $\theta_4$--$\theta_7$) using the  \cite{Erevetal17} algorithm, for 749{,}700 pairs of lotteries.}
\label{tab:extratab3-4-combined}
\end{table}

\newpage

\subsection{When our test fails to detect non-EU}\label{app:fail-test}
All seven prospect theory ground truths we consider in the anomaly generation procedure in Section \ref{sec:MR} fail our test from Section \ref{sec:logit}. Figure \ref{fig:prospect_soln} plots the parameter values for which the test is satisfied.\footnote{As illustrated by the figure, there is a measure zero set of parameters for which our test fails to reject EU. To avoid this, one can repeat our test (tweaking the lotteries slightly) so that the corresponding curves in Figure \ref{fig:prospect_soln} shift.}

\begin{figure}[h!]
    \centering
    \includegraphics[width=10cm]{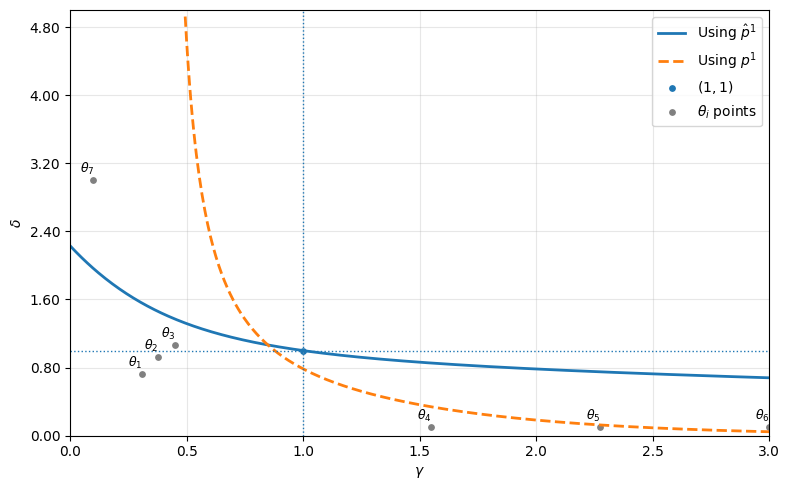}
    \caption{Values of $(\delta, \gamma)$ at which the  utility difference condition $U_{(\delta, \gamma)}(p) - U_{(\delta, \gamma)}(q^1) = U_{(\delta, \gamma)}(p^2) - U_{(\delta, \gamma)}(q^2)$ under the prospect theory ground truth. The blue line corresponds to $p = \hat p^1$, and the orange line corresponds to $p = p^1$.}
    \label{fig:prospect_soln}
\end{figure}

\subsection{A stronger version of Lemma \ref{lemma_convex_closure}}\label{app:proofsplus}

\begin{lemma} \label{lemma_prelim_SI2-strong}
    \begin{align}
     G = \left\{ y \in \{0,1\}^n :  -|J^-| \leq \sum_{j \in J^+} y_j - \sum_{j \in J^-} y_j - y_i \leq |J^+| - 1   \right\}. \nonumber
    \end{align}
\end{lemma}
\begin{proof}
Let $y \in \{0,1\}^n$ satisfy $-|J^-| \leq \sum_{j \in J^+} y_j - \sum_{j \in J^-} y_j - y_i \leq |J^+| - 1$.   Suppose, for a contradiction, that $y \in R$. We then see that
\begin{align}
    \sum_{j \in J^+} y_j - \sum_{j \in J^-} y_j - y_i = \begin{cases}
       |J^+|    & \text{if } y \in R_1 \\
        - |J^-| - 1  & \text{if } y \in R_2,  
    \end{cases} \nonumber
\end{align}
a contradiction of the inequalities in either case. This proves that RHS $\subseteq G$.

Now suppose $y \in G$, and assume, for a contradiction, that $y \not \in$ RHS. This implies that either
\begin{align}
    \sum_{j \in J^+} y_j - \sum_{j \in J^-} y_j - y_i \leq   -|J^-|-1 \hspace{10pt} \text{or} \hspace{10pt} \sum_{j \in J^+} y_j - \sum_{j \in J^-} y_j - y_i  \geq |J^+| \nonumber
\end{align}

The former scenario can only happen if $y_j = 0$ for all $j \in J^+$, $y_j = 1$ for all $j \in J^-$, and $y_i = 1$. This is a contradiction, since all such points are in $R_2$.  The latter scenario can only happen if $y_j = 1$ for all $j \in J^+$, $y_j = 0$ for all $j \in J^-$ and $y_i = 0$. This is a contradiction, since all such points are in $R_1$. This proves $G \subseteq RHS$, as desired. 
\end{proof}
 
Define the vector $ w_{i, J} \in \{-1,0,1\}^n$ such that each index is given by
\begin{align}
    (w_{i, J})_k = \begin{cases}
        -1 & \text{if } k \in J^- \cup \{i\} \\
        1 & \text{if } k \in J^+ \\
        0 & \text{else}
    \end{cases}. \nonumber
\end{align}
The set $G=\{0,1\}^n \setminus R$ in the previous lemma can be rewritten using matrix notation as
\begin{align}
    G & = \{ y \in \{0,1\}^n : -|J^-| \leq  \langle y, w_{i, J} \rangle \leq |J^+| - 1  \} \nonumber \\
    & = \Bigg\{ y \in \{0,1\}^n :  \underbrace{\begin{pmatrix}
        w_{i, J}^T \\ 
        - w_{i, J}^T
    \end{pmatrix}}_{\equiv A_{i, J}} y \geqq \begin{pmatrix}
        -|J^-| \\ 1- |J^+|
    \end{pmatrix}  \Bigg\},  \nonumber
\end{align}
where $\geqq$ refers to element-wise $\geq$ and $A_{i, J} \in \{-1,0, 1\}^{2 \times n}$.

We say that a matrix $A \in \bR^{k \times n}$ is totally unimodular if every square sub-matrix $A'$ of $A$  has $\det(A') \in \{-1,0,1\}$.

\begin{lemma}\label{lem:TU}
    $A_{i, J}$ is totally unimodular.
\end{lemma}
\begin{proof}
    Since  $A_{i, J} \in \{-1,1\}^{2 \times n}$, it only has $1 \times 1$ and $2 \times 2$ square sub-matrices. That the determinant of all $1 \times 1$ submatrices is in $\{-1,0,1\}$ is immediate. As for the $2 \times 2$ case, observe that for any $a,b \in \bR$,
    \begin{align}
        \det \begin{pmatrix}
            a & b \\
            -a & -b
        \end{pmatrix} =  -ab - (-ab) = 0, \nonumber 
    \end{align}
    as desired.
\end{proof}

We will make use of the following mathematical fact. For any set $X$ let $\ext(X)$ be the set of extreme points of $X$. 

\begin{lemma}[Hoffman-Kruskal]\label{lem:HK}
Suppose that the convex set $\Phi$ is given by     \begin{align}
         \Phi = \left\{ \phi \in \bR^n : A \phi \geqq b, \,\, c \leqq \phi \leqq d \right\}
    \end{align}
    for some     $b,c,d  \in \bZ^k$ and a totally unimodular matrix $A \in \bZ^{k \times n}$. Then $\ext(\Phi)\subseteq \bZ^n$. 
\end{lemma}

\begin{proof}
    Because $\Phi$ is a rational polyhedron, by Theorem 4.1 of \cite{integer_programming} the set of vertices of $\Phi$ equals to $\ext(\Phi)\subseteq \bZ^n$ if and only if $\Phi$ is \textit{integral}, meaning that $\Phi=\conv(\Phi\cap \bZ^n)$. By Theorem 4.5 of \cite{integer_programming}, the set $\Phi$ is integral. 
\end{proof}

\begin{lemma} \label{lemma_convex_closure-strong}
    \begin{align}
        \conv(G) = \underbrace{\left\{ \phi \in [0,1]^n : -|J^-| \leq \langle \phi , w_{i, J} \rangle\leq |J^+| -1     \right\}}_{:= \Phi}  \nonumber
    \end{align}
\end{lemma}
\begin{proof}   
    That $\conv(G) \subseteq \Phi$ is immediate, since every $y \in G$ satisfies the constraint $-|J^-| \leq \langle y, w_{i, J} \rangle \leq |J^+| -1$, and the convex combination of any such $y$ continues to satisfy the constraint.  To prove $\Phi \subseteq \conv(G)$, it suffices to show  $\ext(\Phi) \subseteq G$.\footnote{This is because $\ext(\Phi) \subseteq G$ implies $\conv(\ext(\Phi)) \subseteq \conv(G)$, and  $\conv(\ext(\Phi)) = \Phi$ for any compact convex set  $\Phi$ in $\mathbb{R}^n$.} First, we show that $\ext(\Phi) \subseteq \{0,1\}^n$.  This follows from Lemmas \ref{lem:TU} and \ref{lem:HK}, since we can rewrite $\Phi$ as  
    \begin{align}
        \Phi & = \left \{ \phi \in \bR^n : A_{i,J} \phi \geqq \begin{pmatrix}
            -|J^-| \\ 1 - |J^+|
        \end{pmatrix}, \,\,  0 \leqq \phi \leqq 1  \right\},\nonumber 
    \end{align}
    so $\ext(\Phi)\subseteq \bZ^n$. Thus, $\ext(\Phi) \subseteq \bZ^n \cap [0,1]^n = \{0,1\}^n$.     It remains to show that $\ext(\Phi)\subseteq \{0, 1\}^n\setminus R$. This follows because any point in $R$ does not satisfy the inequality constraints (i.e., $\phi \in \ext(\Phi)$ implies $\phi \not \in R$).
\end{proof}

In other words, the facets of the polyhedron $\conv(G)$ are defined by the same linear constraints $-|J^-| \leq \langle \phi , w_{i, J}\rangle \leq |J^+| - 1$ that define the set $G$ itself (in addition to the box constraints $[0,1]^n$).

\end{document}